\documentclass[11pt]{article}%
\usepackage{amssymb, amsfonts, amsmath}
\usepackage{graphicx}
\usepackage{float}
\usepackage{geometry}%
\usepackage{mathtools}
\usepackage{nicefrac}
\usepackage{tabularx, booktabs, multirow}
\usepackage{epstopdf}
\usepackage{hyperref}
\usepackage{multicol}
\usepackage[natbib, style=apa]{biblatex}
\usepackage[dvipsnames]{xcolor}
\usepackage[ruled, vlined]{algorithm2e}  
\usepackage{subcaption}

\usepackage{pifont} 
\newcommand{\cmark}{\ding{51}}%
\newcommand{\xmark}{\ding{55}}%

\providecommand{\U}[1]{\protect\rule{.1in}{.1in}}
\newtheorem{theorem}{Theorem}[section]

\newtheorem{corollary}{Corollary}[section]

\newtheorem{definition}{Definition}[section]

\newtheorem{lemma}{Lemma}[section]

\newtheorem{problem}{Problem}[section]

\newtheorem{remark}{Remark}[section]

\newenvironment{proof}[1][Proof]{\noindent\textbf{#1.} }{\ \rule{0.5em}{0.5em}}
\numberwithin{equation}{section}

\numberwithin{equation}{section}

\newcommand{\be}{\begin{equation}}
\newcommand{\ee}{\end{equation}}

\newcommand{\bq}{\begin{eqnarray}}
\newcommand{\eq}{\end{eqnarray}}

\begin{document}

\title{Exploratory Mean-Variance Portfolio Optimization with Regime-Switching Market Dynamics}
\author{Yuling Max Chen \thanks{Department of Statistics and Actuarial Science, University of
Waterloo, Waterloo ON N2L 3G1, Canada (yuling.chen@uwaterloo.ca)}
\and Bin Li\thanks{Department of Statistics and Actuarial Science, University of
Waterloo, Waterloo ON N2L 3G1, Canada (bin.li@uwaterloo.ca)}
\and David Saunders \thanks{Department of Statistics and Actuarial Science, University of
Waterloo, Waterloo ON N2L 3G1, Canada (dsaunder@uwaterloo.ca).}}
\date{{\small \today}}
\maketitle

\begin{center}
{\large \bf Abstract}
\end{center}

Considering the continuous-time Mean-Variance (MV) portfolio optimization problem, we study a regime-switching market setting and apply reinforcement learning (RL) techniques to assist informed exploration within the control space. We introduce and solve the {\bf E}xploratory {\bf M}ean {\bf V}ariance with {\bf R}egime {\bf S}witching (EMVRS) problem. We also present a Policy Improvement Theorem. 
Further, we recognize that the widely applied Temporal Difference (TD) learning is not adequate for the EMVRS context, hence we consider Orthogonality Condition (OC) learning, leveraging the martingale property of the induced optimal value function from the analytical solution to EMVRS. We design a RL algorithm that has more meaningful parameterization using the market parameters and propose an updating scheme for each parameter. Our empirical results demonstrate the superiority of OC learning over TD learning with a clear convergence of the market parameters towards their corresponding ``grounding true" values in a simulated market scenario. In a real market data study, EMVRS with OC learning outperforms its counterparts with the highest mean and reasonably low volatility of the annualized portfolio returns. 



\vspace{-0.5cm}


\newpage

\section{Introduction} \label{sec:introduction}

Mean-Variance (MV) portfolio optimization has been widely studied since being introduced in \citet{markowitz1952portfolio}. 
In the continuous-time setting, an investor with the classical MV objective aims to simultaneously maximize their expected portfolio value and minimize the portfolio volatility at the end of the investment horizon. This has been studied as a stochastic linear-quadratic problem in \citet{zhou2000continuous}, followed by \citet{chiu2006asset} who solved for the MV efficient frontier, and \citet{xie2008continuous} who derived the optimal investment policy under an incomplete market setting; see \citet{zhang2018portfolio, kalayci2019comprehensive} for a broader overview of the the past work in the MV literature.  

Amongst many variants of the MV problem, this paper focuses on the regime-switching setting. 
Past work such as \citet{maheu2000identifying, ang2002international, ang2004regimes} have demonstrated that incorporating a regime-switching model better fits real market data and that ignoring market regimes in investment has a cost. 
\citet{yin2004markowitz} considered a discrete-time MV problem with regimes modelled on an aggregated process and asymptotically derived the optimal portfolio selection strategy from the optimal solution of its continuous-time counterpart. 
\citet{wu2011multi, wu2014multi} solved regime-switching MV portfolio allocation problems with an uncertain investment horizon, where the time to exit the market is modeled by a conditional distribution given the market regime. 




Recent research in stochastic control \citep{wang2020continuous, jiang2022reinforcement, jia2022policy, jia2023q, dai2023learning, denkert2024control, wu2024reinforcement} 
has shown the advantage of adopting RL techniques into classical stochastic control problems. The basic strategy is to employ the Stochastic RL Algorithm introduced in \citet{gullapalli1990stochastic}, that replaces the optimal solution to the classical stochastic control problem with a probability distribution. By introducing stochasticity to the optimal control solution, we enable exploration within the control space, while maintaining exploitation towards optimality. This allows the induced optimal policy distribution to be more robust against the randomness of the dynamic environment. \citet{wang2020reinforcement} introduced the ``exploratory extension" of the portfolio value dynamic and solved for the corresponding exploratory optimal solution to the linear-quadratic problem, where they proved the asymptotic equivalence between the classical solution and the exploratory solution. Later, \citet{wang2020continuous} introduced the {\it Exploratory Mean-Variance (EMV)} problem and solved for the optimal investment strategy, which is a probability distribution over the control space rather than a deterministic control function. 


In this paper, we study the RL-facilitated EMV problem by \citet{wang2020continuous} in 
the context of a regime-switching process, named as the {\bf E}xploratory {\bf M}ean {\bf V}ariance with {\bf R}egime {\bf S}witching (EMVRS) problem. 
The continuous-time MV portfolio optimization problem with regime-switching 
but without reinforcement learning has previously 
been studied by \citet{zhou2003markowitz}. 
On top of that, we design a value-based RL algorithm that reparameterizes the analytical solution to the EMVRS problem as the value function. In contrast to defining the value function as a randomly initialized neural network or polynomial approximator such as in \citep{jiang2017deep, wen2021reinforcement, wu2024reinforcement}, we configure a RL model whose parameters all have practical meanings. 
Finally, we applied the martingale property of the deduced optimal EMVRS value function, specifically the {\it Orthogonality Condition} studied by \citet{JMLR:v23:21-0947}, to the updating scheme of the parameters in the RL algorithm. 
In contrast to TD learning, this achieves convergence to the true parameters in our simulation study.

\citet{wu2024reinforcement} also studied the RL-facilitated MV portfolio selection problem with a regime switching setting. \citet{wu2024reinforcement} formulated the problem as a Partially Observable Markov Decision Process with two unobservable market regimes, estimating the regimes using the  
 the Wonham filter (Eq. (10) of \citet{wu2024reinforcement}). 
There three major differences between our work and \citet{wu2024reinforcement}. Firstly, \citet{wu2024reinforcement} 
assume that the volatility is regime-independent, while we 
allow for regime-dependent volatility. Secondly, the RL algorithm in \citet{wu2024reinforcement} requires the market 
parameters\footnote{The parameters of the market dynamics and regime dynamics correspond to $\Tilde{\sigma}, \Tilde{mu}_i, \Tilde{\lambda}_i$ for $i = 1,2$ in Algorithm 1 of \citet{wu2024reinforcement}.} to be either given or estimated from the data, and selects the Martingale Loss (first introduced by \citet{JMLR:v23:21-0947}) to update the RL model parameters. We notice from \citet{JMLR:v23:21-0947} that the Martingale Loss requires the knowledge of the true market parameters, which is unavailable in the real market. 
According to the Mean-Blur problem arguments in \citet{luenberger2013investment},
it is possible that replacing the market parameters with their estimates can cause the RL model to update towards an erroneous target. 
Therefore, we choose the Orthogonality Condition Loss (also introduced by \citet{JMLR:v23:21-0947}), that does not depend on the true market parameters while simultaneously applies the martingality properties of our analytical solutions. As a result, our RL algorithm does not require the knowledge of the market parameters, instead, it learns the market parameters through training. Thirdly, by parameterizing the RL model with the market parameters, we show in our simulation study that the parameters of our proposed RL model can converge to their corresponding true values, with randomly chosen starting point.

\bigskip 

The remainder of the paper is structured as follows. In 
Section \ref{sec:preliminary}, we present background material and the problem setup. Then, we introduce the EMVRS problem, provide an analytical solution, and establish a Policy Improvement Theorem in Section \ref{sec:methodology}. 
In Section \ref{sec:algorithm}, 
we present a reparameterization of the EMVRS problem, and 
various algorithms for updating the market parameters. 
We present numerical results on simulated and real market data in Section \ref{sec:simulation}.\footnote{The EMVRS model is developed in Python and the source code is available on \href{https://github.com/MaxGniluynehc/EMVRS/}{GitHub}.} 
Section \ref{sec:conclusion} 
presents conclusion and directions for future research.

\section{Preliminaries} \label{sec:preliminary}

In this section, we review the Lagrangian dual formulation of  Markowtiz's Mean-Variance (MV) Problem with a regime-switching market, followed by a stochastic control solution to this problem that does not involve any RL techniques. This was originally done by \citet{zhou2003markowitz}. The Lagrangian dual formulation of the MV problem was also applied in \citet{wang2020continuous}, although they did not consider regime-switching. In the successive sections, we refer to the problem of \citet{zhou2003markowitz} as the {\it Mean-Variance with Regime-Switching (MVRS) problem}, which is fundamental to our proposed method. 

\subsection{Problem Formulation} \label{sec:problem-formulation}
Consider an investor who manages a portfolio with an investment horizon $T > 0$. For ease of presentation, the market is simplified to consist of one risky asset, the stock $\{S_t\}_{t\in[0,T]}$, and one risk-free asset, the bond $\{B_t\}_{t\in[0,T]}$. Denote by $\{W_t\}_{t \in [0,T]}$ a one-dimensional Brownian Motion defined on a filtered probability space $(\Omega, \mathcal{F}, \{\mathcal{F}_t\}_{t\in[0,T]}, \mathbb{P})$ that satisfies the usual conditions. 

Recognizing that the market has ``good" and ``bad" states, we further denote $\alpha_t$ as the regime of the market at time $t \in [0,T]$. For any time $t \in [0,T]$, $\alpha_t$ takes a value from the set $\{1,\cdots, l\}$, where $l$ is the total number of the market regimes. 
We model the market regimes using a continuous-time, adapted, stationary and time-homogeneous Markov Chain with transition matrix $P = (p_{ij}(t))_{i = 1, \cdots, l}^{j = 1, \cdots, l}$, where $p_{ij}(t) := \mathbb{P}(\alpha_t = j | \alpha_0 = i)$, the probability of transitioning from regime $i$ at time 0 to regime $j$ at time $t$. The Markov Chain generator $Q = (q_{ij})_{i = 1, \cdots, l}^{j = 1, \cdots, l}$ is defined as $q_{ij} = \lim_{t\to 0+} t^{-1}(p_{ij}(t) - \delta_{ij})$ with $\delta_{ij}$ the Kronecker delta. It is so-defined such that the transition matrix is the matrix exponential of the generator, i.e., $P(t) = e^{tQ}$. Moreover, we assume that the regime's Markov Chain is independent of the Brownian Motion $W$. 

The dynamics of the stock and the bond are driven by two stochastic processes 
\begin{align}
    & dS_t = S_t \left\{\mu(t, \alpha_t) dt + \sigma(t, \alpha_t dW_t \right\}, \text{ with } S_0 > 0 \label{eq:problem-formulation-stock-dynamics} \\
    & dB_t = r(t, \alpha_t) B_t dt, \text{ with } B_0 > 0 \label{eq:problem-formulation-bond-dynamics}
\end{align}
where $\mu(t,i) \in \mathbb{R}$ and $\sigma(t,i) \in \mathbb{R}_+$ are respectively the mean and volatility of the stock return and $r(t,i)\in \mathbb{R}_+$ is the risk-free interest rate, at time $t \in [0,T]$ in market regime $i \in \{1,\cdots, l\}$. 

At each time $t \in [0,T]$, denote $X_t$ as the investor's portfolio value. The investor reallocates their portfolio by investing the amount $u_t$ in the stock and $X_t - u_t$ in the bond. Under the self-financing assumption, the {\it portfolio value process} can be derived as
\begin{equation} \label{eq:problem-formulation-wealth-process-classic}
 \begin{split}
    d X^u_t & = \left[r(t, \alpha_t) X^u_t + [\mu(t, \alpha_t) - r(t, \alpha_t)]u_t \right] dt + \sigma(t, \alpha_t) u_t dW_t \\
    & = \left[r(t, \alpha_t) X^u_t + \rho(t, \alpha_t) \sigma(t, \alpha_t) u_t \right] dt + \sigma(t, \alpha_t) u_t dW_t
\end{split} 
\end{equation}
given the initial portfolio value $X_0 = x_0 > 0$ and the initial regime $\alpha_0 = i_0 \in \{1, \cdots, l\}$. Here, $\rho(t, \alpha_t) := \sigma^{-1}(t,\alpha_t)(\mu(t, \alpha_t) - r(t, \alpha_t))$ is the Sharpe ratio, and $\{X^u_t\}_{t\in [0,T]}$ with a superscript $u$ represents the portfolio value process that follows the {\it control policy} $u := \{u_t\}_{t\in [0,T]}$. We next define admissible control policies and state the classical Markowitz MV Problem.

\begin{definition} \label{def:admissible-control-space}
    We say that a control policy $u(t)$ is {\it admissible}, denoted as $u(t) \in \mathcal{A}$, if it satisfies the following conditions:  
    \begin{enumerate}
        \item[(i)] $u:[0,T] \mapsto \mathbb{R} \in L^2_{\mathcal{F}}(0,T)$, i.e., $u_\cdot$ is a $\mathbb{R}$-valued, $\{\mathcal{F}_t\}_{t\in [0,T]}$-adpated function such that $\mathbb{E} \left[\int_0^T |u_t|^2 dt \right] < +\infty$;
        \item[(ii)] the SDE (\ref{eq:problem-formulation-wealth-process-classic}) has a unique solution $X^u_\cdot$ corresponding to $u_\cdot$. 
    \end{enumerate}
\end{definition}

\begin{problem}[Classical Markowitz MV Problem] \label{pb:classical-MV-problem-origin}
\begin{equation}\label{eq:classical-MV-problem-origin}
    \min_{u \in \mathcal{A}} Var(X^u(T)) \text{ subject to } \mathbb{E}(X^u(T)) = z
\end{equation}
where $z > 0$ is a prespecified target for the expected terminal wealth, i.e., the expected wealth to be achieved at the end of the investment horizon.
\end{problem}

Noticing that the classical MV problem is a constrained optimization problem, 
we consider the Lagrangian dual of Problem \ref{pb:classical-MV-problem-origin}. 
Moreover, under the regime-switching market setting, we specify that the expectation in the original MV problem is actually conditioned on the initial states, including the investor's initial portfolio value $x_0 > 0$ and the initial market regime $i_0 \in \{1,\cdots,l\}$.

\subsection{The MVRS Problem}
Following the standard Dynamic Programming Principle (DPP) arguments, \citet{zhou2003markowitz} proposed and solved the Mean-Variance with Regime Switching (MVRS) problem (Problem \ref{pb:classical-MV-problem-lagrange}), which is the Lagrangian dual of Problem \ref{pb:classical-MV-problem-origin}. This modification reforms a time-inconsistent stochastic control problem into a time-consistent one, which makes it possible to obtain a precommitted solution. 
\begin{problem}[Mean-Variance with Regime Switching (MVRS) Problem] \label{pb:classical-MV-problem-lagrange}
\begin{equation}\label{eq:classical-MV-problem-lagrange}
    \min_{u \in \mathcal{A}} \Bigl\{ J(u, x_0, i_0; \lambda) := \mathbb{E}[(X^u_T + \lambda - z)^2 | X^u_0 = x_0, \alpha_0 = i_0] - \lambda^2 \Bigr\}
\end{equation}
where $\lambda > 0$ is the Lagrangian multiplier, $z > 0$ is a prespecified target for the expected terminal wealth and $J(u, x_0, i_0; \lambda)$ is the objective function. 
\end{problem}

Define the {\it value function} as the infimum of the objective function over all admissible controls $u\in\mathcal{A}$. For $0 \leq t < s \leq T, X_t = x, \alpha_t = i$, and $u \in \mathcal{A}$, the {\it value function} is given by
\begin{align}
    V(t,x,i) & = \inf_{u \in \mathcal{A}} \mathbb{E}[(X^u_T + \lambda -z)^2 | X_t = x, \alpha_t = i] - \lambda^2  \\
    & = \inf_{u \in \mathcal{A}} \mathbb{E}[V(s, X^u_s, \alpha_s) | X_t = x, \alpha_t = i].
\end{align}
Following DPP arguments, the value function can be deduced to satisfy the Hamilton-Jacobi-Bellman (HJB) Equation
\begin{equation} \label{eq:MVRS-HJB-origin}
\begin{split}
    & v_t(t,x,i) + \sum_{j=1}^2 q_{ij} v(t,x,j) + v_x(t,x,i) r(t,i)x \\
    & + \min_u \left\{v_x(t,x,i) \rho(t,i) \sigma(t,i) u(t,x,i) + \frac{1}{2} v_{xx}(t,x,i) \sigma^2(t,i) u^2(t,x,i) \right\} = 0
\end{split}
\end{equation}
Solving this HJB for the optimal control $u^*(\cdot)$ (i.e., investment policy) gives: 
\begin{equation} \label{eq:MVRS-optimal-control-0}
    u^*(t,x,i) = -\frac{\rho(t,i)v_x(t,x,i)}{\sigma(t,i)v_{xx}(t,x,i)}
\end{equation}
Substituting this optimal control $u^*(\cdot)$ back into Eq. \ref{eq:MVRS-HJB-origin}, \cite{zhou2003markowitz} solve for the corresponding optimal value function from the HJB equation and derive the explicit form of $u^*(\cdot)$, which we summarize in the following theorem.

\begin{theorem} \label{thm:MVRS-solution}
Problem \ref{pb:classical-MV-problem-lagrange} has an optimal control
\begin{equation} \label{eq:MVRS-optimal-control}
    u^*(t,x,i) = -\frac{\rho(t,i)}{\sigma(t,i)} (x + (\lambda - z) H(t,i))
\end{equation}
and the corresponding value function is given by
\begin{equation} \label{eq:MVRS-optimal-value-function}
\begin{split}
    & \inf_{u \in \mathcal{A}} J(u_\cdot; x_0, i_0, \lambda) \\
    & = \mathbb{E} [(X^{u^*}_T + \lambda - z)^2 | X^u_0 = x_0, \alpha_0 = i_0] \\
    & = V(0, x_0, i_0) \\
    & = P(0, i_0) [x_0 + (\lambda-z) H(0, i_0)]^2 \\ 
    & + (\lambda-z)^2 \mathbb{E} \left[\left. \int_0^T \sum_{j=1}^2 q_{\alpha(s) j} P(s,j) (H(s,j) - H(s, \alpha(s)))^2 ds \right| \alpha_0 = i_0 \right] 
    - \lambda^2,
\end{split}
\end{equation}
where $P(t,i)$ and $H(t,i)$ are the solutions to the following two Ordinary Differential Equations (ODE)
\begin{align}
& \begin{cases} \label{eq:MVRS-P(t,i)-ODE}
    \dot{P}(t,i) = (\rho^2(t,i) - 2r(t,i)) P(t,i) - \sum_{j=1}^l q_{ij}P(t,j) \\
    P(T,i) =1, \text{ for } i =1,2,\cdots,l
\end{cases} \\
& \begin{cases} \label{eq:MVRS-H(t,i)-ODE}
    \dot{H}(t,i) = r(t,i) H(t,i) - \frac{1}{P(t,i)} \sum_{j=1}^l q_{ij} P(t,j) (H(t,j) - H(t,i)) \\
    H(T,i) =1, \text{ for } i =1,2,\cdots,l
\end{cases}
\end{align}
\end{theorem}

While \cite{zhou2003markowitz} considered the market regime dynamics that affect investors' trading activity, the optimal control $u^*(t,x,i)$ (Eq. \ref{eq:MVRS-optimal-control}) is an $\mathbb{R}$-valued, deterministic function given the time $t$, portfolio value $x$ and regime $i$. This indicates a precommitted policy, which is fixed and does not explore within the control space once it is considered to be optimized. However, recent work has demonstrated that a RL-facilitated exploration within the control space can achieve better performance, although under a non-Regime-Switching setting (\citet{wang2020continuous}). This motivates us to consider a RL-extension to Problem \ref{pb:classical-MV-problem-lagrange}.

 
\section{Exploratory Mean-Variance with Regime Switching Problem} \label{sec:methodology}
In this section, we make an RL extension to the aforementioned MVRS problem (Problem \ref{pb:classical-MV-problem-lagrange}) in \citet{zhou2003markowitz}, following the exploratory policy formulation from \citet{wang2020continuous}. We hereafter refer our problem to the {\it Exploratory Mean-Variance with Regime-Switching (EMVRS) problem}. This is a generalized problem of the past works \citep{zhou2003markowitz, wang2020continuous}, which accounts for the dynamics of the market regimes and allows for informed exploration within the control space. 

\subsection{Extension to the MVRS Problem} \label{sec:EMVRS-problem-formulation}

We adopt from Section \ref{sec:problem-formulation} the stock and bond dynamics (Eq. \ref{eq:problem-formulation-stock-dynamics} and \ref{eq:problem-formulation-bond-dynamics}) and the Markov-modulated regime-switching setting. Further, in order to allow for exploration within the control space, we 
define the exploratory control as a probability distribution, called the {\it policy distribution}. For any subset of real numbers, $\mathcal{V} \subseteq \mathbb{R}$, we denote $\mathcal{P}(\mathcal{V})$ as the set of probability density functions defined over $\mathcal{V}$. 


\begin{definition} \label{def:policy-distribution}
    The {\it policy distribution}\footnote{\cite{wang2020continuous} defined this in a similar way and called it the {\it exploratory control}.}, $\boldsymbol{\pi} :=  \{\pi_t(\cdot|\alpha_t)\}_{t\in[0,T]}$, is a distribution-valued, $\{\mathcal{F}_t\}_{t\in[0,T]}$-adapted random variable, where each $\pi_t(\cdot|\alpha_t) : \mathcal{U} \mapsto \mathcal{P}(\mathcal{U})$ is a probability density function over the feasible set of control values $\mathcal{U} \subseteq \mathbb{R}$ conditioned on the regime $\alpha_t \in \{1, \cdots, l\}$, satisfying $\int_{\mathcal{U}} \pi_t(u|\alpha_t) du = 1$ and $\pi_t(u|\alpha_t) \geq 0, \forall t \in [0,T], u \in \mathcal{U}$. To avoid ambiguity, we call $\pi_t$ the policy distribution at time $t\in[0,T]$. 
\end{definition}
\begin{remark}
    The feasibility of control values depends on market regulations and institutional constraints. For example, if short selling is forbidden, then $u \geq 0$ which implies $\mathcal{U} \equiv \mathbb{R}_+$.
\end{remark}

Then, the portfolio value process is extended from Eq. \ref{eq:problem-formulation-wealth-process-classic} to
\begin{align} 
    \begin{split} \label{eq:EMVRS-wealth-process-exploratory}
    d X^\pi_t & = \left[r(t, \alpha_t) X^\pi_t + \int_{\mathcal{A}} \rho(t, \alpha_t) \sigma(t, \alpha_t) \cdot u \cdot \pi_t(u|\alpha_t) du \right] dt + \left(\sqrt{\int_\mathcal{A} \sigma^2(t, \alpha_t) u^2 \cdot \pi_t(u|\alpha_t)  du} \right) dW_t
    \end{split}
\end{align}
We include a brief derivation of this equation in the Appendix.

\begin{definition}
   For any time $t\in[0,T]$, we say that $\pi_t$ is an {\it admissible policy distribution}, denoted as $\pi_t \in \mathcal{A}^\pi$, if the following conditions are satisfied: 
    \begin{enumerate}
        \item[(i)] For $\alpha_t \in \{1,\cdots, l\}$, $\pi_t(\cdot|\alpha_t)$ is a policy distribution as described in Definition \ref{def:policy-distribution}. 
        \item[(ii)] For any feasible set of control values $\mathcal{U} \subseteq \mathbb{R}$,
        the stochastic process $\left\{\int_{\mathcal{U}} \pi_t(u|\alpha_t) du\right\}_{t \in [0,T]}$ is $\{\mathcal{F}_t\}_{t\in[0,T]}$-adapted. 
        \item[(iii)]
        For all $(t, i)\in [0,T) \times \{1,\cdots, l\}$,
        \begin{equation}
            \mathbb{E} \left[\left. \int_t^T \int_\mathcal{A} u^2 \pi_s(u|\alpha_s) du ds \right| \alpha_t = i\right] < \infty
        \end{equation}
        This implies that the exploratory wealth dynamic SDE (Eq. \ref{eq:EMVRS-wealth-process-exploratory}) has a unique solution $X^\pi_t$ corresponding to $\pi_t(\cdot|\alpha_t)$, according to the arguments in \citet{zhou2003markowitz, dai2023learning}. 
        \item[(iv)] For all $(t,x, i)\in [0,T) \times \mathbb{R} \times \{1,\cdots,l\}$ and constant $\xi < \infty$, 
        \begin{equation}
            \mathbb{E} \left[\left. (X^\pi_T +\lambda - z)^2 + \xi \int_t^T \int_\mathcal{A} \pi_s(u|\alpha_s) \log \pi_s(u|\alpha_s) du ds \right| X^\pi_t = x, \alpha_t = i \right] < \infty
        \end{equation} 
    \end{enumerate}
    If $\pi_t\in \mathcal{A}^\pi, \forall t\in[0,T]$, we denote $\boldsymbol{\pi} \in \mathcal{A}^\pi$.
\end{definition}

We propose the EMVRS problem as an entropy-regularized version of the MVRS problem (\ref{pb:classical-MV-problem-lagrange}), which restrains the policy distribution within a certain family of distributions. 
\begin{problem}[Exploratory Mean-Variance with Regime-Switching (EMVRS) Problem] \label{pb:EMVRS-problem}
\begin{equation}
    \min_{\pi \in \mathcal{A}^\pi} \mathbb{E}\left[ \left. (X^\pi_T + \lambda - z)^2 + \xi \int_0^T \int_\mathcal{A} \pi_t(u|\alpha_t) \log \pi_t(u|\alpha_t) du dt \right| X^\pi_0 = x_0, \alpha_0 = i_0 \right] - \lambda^2 
    \end{equation}
where $\xi > 0$ is the exploration weight, $x_0 > 0$ and $i_0 \in \{1, \cdots, l\}$ are respectively the portfolio value and market regime at $t=0$. 
\end{problem}
For any $(t, x) \in [0,T] \times \mathbb{R}$, $i \in \{1,\cdots, l\}$ and admissible policy distribution $\pi \in \mathcal{A}^\pi$, we define the {\it value function} (corresponding to $\pi$) $V^\pi$ and the {\it optimal value function} $V^*$ as
\begin{align} 
    V^\pi (t,x,i) & := \mathbb{E}\left[\left. (X^\pi_T + \lambda - z)^2 + \xi \int_t^T \int_\mathcal{A} \pi_k(u|\alpha_k) \log \pi_k(u|\alpha_k) du dk \right| X^\pi_t = x, \alpha_t = i\right] - \lambda^2  \label{eq:EMVRS-value-function} \\
    V^*(t,x,i) & := \inf_{\pi \in \mathcal{A}^\pi} V^\pi(t,x,i) \label{eq:EMVRS-DPP-recursive-form-of-Value-function}
\end{align}
By Bellman's Principle of Optimality, we can derive the recursive form of the optimal value function, for $0 \leq t \leq s \leq T, x\in\mathbb{R}, i \in \{1, \cdots, l\}:$
\begin{equation} \label{eq:EMVRS-DPP-recursive-form-of-Value-function}
\begin{split}
    V^*&(t,x,i) = \inf_{\pi \in \mathcal{A}^\pi}  \mathbb{E}\left[\left. V^\pi(s, X^{\pi}_s, \alpha_s) + \xi \int_{t}^s \int_{\mathcal{A}} \pi_k(u|\alpha_k) \log \pi_k(u|\alpha_k) du dk \right| X^\pi_t = x, \alpha_t = i \right].
\end{split}
\end{equation}
Following DPP arguments, we know by assuming $V$ is smooth and applying It\^o's formula to Eq. \ref{eq:EMVRS-DPP-recursive-form-of-Value-function} that the optimal value function $V^*$ satisfies the HJB equation
\begin{equation} \label{eq:EMVRS-HJB-raw}
\begin{split}
    v_t & (t,x, i) + v_x(t,x,i)r(t,i)x + \sum_{j=1}^l q_{i j} v(t,x,j) \\
    & + \min_{\pi(\cdot|i) \in \mathcal{A}^\pi} \left\{ \int_\mathcal{A} \left[\frac{1}{2} v_{xx}(t,x,i) \sigma^2(t, i ) u^2 + v_{x}(t,x,i) \rho(t,i)\sigma(t,i)u + \xi \log \pi(u|i) \right] \pi(u|i) du \right\} = 0.
\end{split} 
\end{equation}
Solving the minimization part in Eq. \ref{eq:EMVRS-HJB-raw} yields the optimal policy distribution $\pi^*$, which is a Gaussian distribution with the mean coinciding with the optimal control of the MVRS problem (Eq. \ref{eq:MVRS-optimal-control-0}). 
\begin{equation} \label{eq:EMVRS-optimal-policy-distribution-raw}
\begin{split}
    \pi^*_t(u; i) & = \frac{\exp\left(-\frac{1}{\xi}[\frac{1}{2} v_{xx}(t,x,i) \sigma^2(t,i)u^2 + v_x(t,x,i)\rho(t,i)\sigma(t,i)u]\right)}{\int_\mathcal{A} \exp\left(-\frac{1}{\xi}[\frac{1}{2} v_{xx}(t,x,i) \sigma^2(t,i)u^2 + v_x(t,x,i)\rho(t,i)\sigma(t,i)u]\right)du} \\
    & = N \left(u \left| -\frac{\rho(t,i)v_x(t,x,i)}{\sigma(t,i) v_{xx}(t,x,i)}, \frac{\xi}{\sigma^2(t,i) v_{xx}(t,x,i)} \right. \right).
\end{split}
\end{equation}
Substituting $\pi^*$ into Eq. \ref{eq:EMVRS-HJB-raw} reduces the HJB equation to 
\begin{equation} \label{eq:EMVRS-HJB}
\begin{split}
    v_t (t,x, i) + v_x(t,x,i)r(t,i)x + \sum_{j=1}^l q_{i j} v(t,x,j) - \frac{1}{2} \frac{\rho^2(t,i) v_x^2(t,x,i)}{v_{xx}(t,x,i)} - \frac{\xi}{2} \log\left(\frac{2\pi \xi}{ \sigma^2(t,i) v_{xx}(t,x,i)} \right) = 0.
\end{split}
\end{equation} 
The rest of the task is to solve the ``reduced" HJB equation (Eq. \ref{eq:EMVRS-HJB}) for the optimal value function $V^*$, which is given in the theorem below. The proof is given in the Appendix.

\begin{theorem} \label{thm:EMVRS-solution}
Problem \ref{pb:EMVRS-problem} has an optimal policy distribution $\boldsymbol{\pi^*} = \{\pi_t^*\}_{t\in[0,T]}$, where each $\pi_t^*$ is given by a Gaussian distribution 
\begin{equation} \label{eq:EMVRS-optimal-policy-distribution}
    \pi_t^*(u; i) = N \left(-\frac{\rho(t,i)}{\sigma(t,i)} [x + (\lambda-z) H(t,i)], \frac{\xi}{2 \sigma^2(t,i) P(t,i)} \right),
\end{equation}
with $(t,x) \in [0,T] \times \mathbb{R}$ and $i \in \{1, \cdots, l\}$. The corresponding optimal value function is given by
\begin{equation} \label{eq:EMVRS-optimal-value-function}
\begin{split}
     V^*(t,x,i) = P(t,i)[x+(\lambda-z)H(t,i)]^2 + (\lambda-z)^2 C(t,i) + D(t,i) - \lambda^2,
\end{split}
\end{equation}
where $P(t,i), H(t,i)$ are solutions of the two ODEs in Eq. \ref{eq:MVRS-P(t,i)-ODE} and \ref{eq:MVRS-H(t,i)-ODE}, and
\begin{align}
    & C(t,i) = \sum_{m=1}^l \sum_{j=1}^l  \int_t^T p_{im}(s-t) q_{mj} P(s,j) (H(s,j) - H(s, m))^2 ds, \\
    & D(t,i) = - \sum_{m=1}^l \int_t^T p_{im}(s-t) \frac{\xi}{2} \log \left(\frac{\pi e \xi}{\sigma^2(s, m) P(s, m)}\right) ds.
\end{align}
Moreover, at initialization with $t=0$, the optimal Lagrange multiplier is 
\begin{equation} \label{eq:EMVRS-optimal-lambda}
    \lambda^* = \frac{z - P(0, i_0) H(0,i_0) x_0}{P(0, i_0) H(0, i_0)^2 + C(0, i_0) - 1} + z.
\end{equation}
\end{theorem}

\begin{remark}
The functions $C(t,i)$ and $D(t,i)$ are solutions to the following two ODEs 
\begin{align}
    & \begin{cases} \label{eq:EMVRS-C(t,i)}
    \dot{C}(t,i) = - \sum_{j=1}^l q_{ij} \left[P(t,j) (H(t,j) - H(t,i))^2 + C(t,j)\right], \\
    C(T,i) = 0, \text{ for } i \in \{1,\cdots, l\},
\end{cases} \\
&\begin{cases} \label{eq:EMVRS-D(t,i)}
    \dot{D}(t,i) = \frac{\xi}{2} \log \left(\frac{\pi \xi}{\sigma^2(t,i) P(t,i)}\right) - \sum_{j=1}^l q_{ij} D(t,j), \\
    D(t,i)= 0, \text{ for } i \in \{1,\cdots, l\}.
\end{cases}
\end{align}
\end{remark}

\begin{remark}
    When there is only one regime, the dynamics of the stock and bond can be reduced from Eq. \ref{eq:problem-formulation-stock-dynamics} and \ref{eq:problem-formulation-bond-dynamics} to: 
    \begin{align}
        & dS_t = S_t \left\{\mu dt + \sigma dW_t \right\}, \text{ with } S_0 > 0 \label{eq:EMV-stock-dynamic} \\
        & dB_t = r B_t dt, \text{ with } B_0 > 0 \label{eq:EMV-bond-dynamic}
    \end{align}
    where $\mu \in \mathbb{R}$ and $\sigma > 0$ are respectively the mean and volatility of the stock returns and $r > 0$ is the risk-free interest rate. Then, Problem \ref{pb:EMVRS-problem} can be simplified to the Exploratory Mean Variance (EMV) problem
    \begin{equation}\label{pb:EMV-problem}
        \min_{\pi \in \mathbb{P}(\mathcal{A})} \mathbb{E}\left[ \left. (X^\pi_T + \lambda - z)^2 + \xi \int_0^T \int_\mathcal{A} \pi_t(u) \log \pi_t(u) du dt \right| X^\pi_0 = x_0 \right] - \lambda^2, 
    \end{equation}
    where $\xi > 0$ is the exploration weight, and $\{X_t\}_{t\in[0,T]}$ is the solution to the SDE 
    \begin{equation}
        dX^\pi_t = \left(rX^\pi_t + \int_\mathcal{A} \rho \sigma u \pi_t(u) du \right) dt + \left(\sqrt{\int_\mathcal{A} \sigma^2 u^2 \pi_t(u) du} \right) dW_t. 
    \end{equation}
    The EMV problem has been addressed in \cite{wang2020continuous}. We summarize their solution in the following corollary. 
    \begin{corollary} \label{thm:EMV-solution}
        The EMV problem has an optimal policy distribution $\boldsymbol{\pi^*} = \{\pi^*_t(\cdot)\}_{t\in [0,T]}$, where each $\pi^*_t:\mathcal{A} \mapsto \mathbb{P}(\mathcal{A})$ is given by the Gaussian distribution
        \begin{equation} \label{eq:EMV-optimal-policy}
            \pi^*_t(u) = 
            N \left(u \left| -\frac{\rho}{\sigma}(x + (\lambda - z)e^{-r(T-t)}),
            \frac{\xi}{2 \sigma^2} e^{(\rho^2-2r)(T-t)} \right.\right), 
            \text{ with } \lambda^* = z-\frac{z e^{(\rho^2-r)(T-t)} - x_0}{e^{\rho^2 T} - 1}
        \end{equation}
        The corresponding value function is given by
        \begin{equation}
            v(t,x, \lambda) = \left(x+(\lambda-z)e^{-r(T-t)}\right)^2 e^{-(\rho^2-2r)(T-t)} 
            + \frac{\xi (\rho^2-2r)}{4}(T^2-t^2) - \frac{\xi}{2} \left[(\rho^2-2r) T - \log \frac{\sigma^2}{\pi\xi}\right] (T-t) - \lambda^2
        \end{equation}
    \end{corollary}
\end{remark}

\subsection{Policy Improvement Theorem} \label{sec:policy-improvement-theorem}
Theorem \ref{thm:EMVRS-solution} provided the optimal policy distribution that solves the EMVRS problem. However, in practice we usually start with an initial guess of the policy $\pi^0$ and iteratively update it, until it converges (or is sufficiently close) to the optimal solution $\pi^*$. In the RL literature, this is known as {\it policy iteration} (\cite{RL-Sutton1998}). In the following Policy Improvement Theorem (PIT), we propose an iteration scheme for policy updates, which is guaranteed to improve or at least not downgrade the current policy. The proof is given in the Appendix. 

\begin{theorem} \label{thm:EMVRS-PIT}
Let $\boldsymbol{\pi} \in \mathcal{A}^\pi$ be any admissible policy distribution and $V^\pi(\cdot, \cdot, i)$ be the corresponding value function as defined in Eq. \ref{eq:EMVRS-value-function}, satisfying $V_{xx}^\pi(t,x,i) > 0$ for any regime $i \in \{1, \cdots, l\}$, time and wealth $(t,x) \in [0,T] \times \mathbb{R}$.
Consider constructing a new policy distribution $\boldsymbol{\pi}^* = \{\pi^*_t\}$, where each $\pi^*_t$ is given by
\begin{equation} \label{eq:PIT-construction-of-pi*}
    \pi^*_t (u; i) = N\left(-\frac{\rho(t,i) V_x^\pi(t,x,i)}{\sigma(t,i) V_{xx}^\pi(t,x,i)}, \frac{\xi}{\sigma^2(t,i) V_{xx}^\pi(t,x,i)} \right).
\end{equation}
If this new policy is admissible, i.e., $\boldsymbol{\pi^*} \in \mathcal{A}^\pi$, then $V^{\pi^*}(t,x,i) \leq V^{\pi}(t,x,i)$,  for all $(t,x) \in [0,T] \times \mathbb{R}, i \in \{1,\cdots,l\}$.
\end{theorem}
According to the PIT, we can always upgrade, or at least not downgrade, the investment policy given the value function. So, the optimality of the policy relies on the optimality of the corresponding value function, which we discuss in more detail in Section \ref{sec:algorithm}. 

\section{RL Algorithm} \label{sec:algorithm}
To optimize the value function, we develop an RL algorithm with the policy updating scheme following the PIT (Theorem \ref{thm:EMVRS-PIT}). We realize that the PIT constructs an improved policy distribution based on the {\it market parameters}, stock return volatility $\sigma(t,i)$ and Sharpe Ratio $\rho(t,i)$ at time $t$ and regime $i$, which are inaccessible in practice. Meanwhile, the optimal value function defined in Theorem \ref{thm:EMVRS-solution} also depends on the market parameters, although indirectly through the $P, H, C, D$ functions. Thus, the optimization of the value function is a matter of learning the market parameters. In other words, if we did know the ``true" market parameters, we could turn off policy exploration (i.e., setting $\xi = 0$) and adopt the optimal investment strategy of the classical MV problem in Eq. (\ref{eq:MVRS-optimal-control-0}). 


For ease of presentation, we hereafter suppose that there are only two regimes, i.e., $l=2$ and $\{\alpha_t\}_{t\in[0,T]} \in \{1,2\}$. We further assume that the market parameters and the interest rates are constant in the same regime, regardless of the time. That is, for any $t_1, t_2 \in [0,T], i \in \{1,2\}$
$$\sigma(t_1,i) = \sigma(t_2,i) =: \sigma_i; \quad \rho(t_1,i) = \rho(t_2,i) =: \rho_i; \quad r(t_1, i) = r(t_2, i) = r_i.$$
So, the value function can be seen as a function of $\theta := (\sigma_1, \sigma_2, \rho_1, \rho_2)$,
\begin{equation} \label{eq:EMVRS-optimal-value-function-reparam}
\begin{split}
     V^\theta(t,x,i) := P^\theta(t,i)[x+(\lambda-z)H^\theta(t,i)]^2 + (\lambda-z)^2 C^\theta(t,i) + D^\theta(t,i) - \lambda^2,
\end{split}
\end{equation}
where $P^\theta, H^\theta, C^\theta, D^\theta$ are the solution to the system of ODEs in Eq. \ref{eq:MVRS-P(t,i)-ODE}, \ref{eq:MVRS-H(t,i)-ODE}, \ref{eq:EMVRS-C(t,i)} and \ref{eq:EMVRS-D(t,i)}. We use the superscript $\theta$ to denote the direct dependence of $P,H,C,D$ on $\theta$ and the indirect dependence of $V^*$ on $\theta$. Denote $\{X_t^\theta\}_{t\in[0,T]}$ as the portfolio value process that follows the investment policy $\boldsymbol{\pi}^\theta$. Then, the PIT-inferred policy distribution $\boldsymbol{\pi}^\theta := \{\pi_t^\theta\}_{t\in[0,T]}$ is reparameterized as
\begin{equation}
\begin{split} \label{eq:PIT-construction-of-pi-reparam}
    \pi^\theta_t (u; i) & := N\left(u \left| -\frac{\rho_i V_x^\theta(t,x,i)}{\sigma_i V_{xx}^\theta(t,x,i)}, \frac{\xi}{\sigma^2_i V_{xx}^\theta(t,x,i)} \right. \right) \\
    & = N\left(u \left| - \frac{\rho_i}{\sigma_i} [x + (\lambda-z) H^\theta(t,i)]^2, \frac{\xi}{2 \sigma_i^2 P^\theta(t,i)}
    \right.\right),
\end{split}
\end{equation}
given $X_t^\theta = x \in \mathbb{R}$ and $\alpha_t = i\in \{1,2\}$.

In the rest of this section, we will introduce two updating schemes for the market parameters --- the Temporal Difference (TD) Learning (Section \ref{sec:TD-learning}) and the Orthogonality Condition Learning (Section \ref{sec:OC-learning}). The TD learning was originally introduced by \citet{sutton1988learning}, and later \citet{wang2020continuous} applied it for the RL training of their EMV problem. More recently, \citet{JMLR:v23:21-0947} argued against the appropriateness of TD learning in stochastic control problems and proposed to leverage the martingality of the value function in the RL training process. We wrap them up with a training algorithm and give the corresponding updating scheme for the market parameters under a simulated market setting in Section \ref{sec:updating-scheme-training-algorithm}. 
We show both theoretically and empirically (in the next section) that the TD Learning is not suitable for the EMVRS problem.
Finally, we conclude this section with some remarks on how our algorithm can be amended to train EMVRS on real market data in Section \ref{sec:training-on-real-data}. 


\subsection{Temporal Difference (TD) Learning} \label{sec:TD-learning}
We know by Bellman's Principle of Optimality that the value function has the recursive form, for $0 \leq t \leq s \leq T, x\in\mathbb{R}, i \in \{1,2\}$
\begin{align} 
    V^*(t,x,i) & = \mathbb{E}\left[\left. V^*(s, X^{\pi}_s, \alpha_s) + \xi \int_{t}^s \int_{\mathcal{A}} \pi_k^*(u|\alpha_k) \log \pi_k^*(u|\alpha_k) du dk \right| X^\pi_t = x, \alpha_t = i \right] \label{eq:DPP-recursive-form-of-V*} \\
    & = \mathbb{E}\left[\left. V^*(s, X^{\pi}_s, \alpha_s) - \xi \int_{t}^s \frac{1}{2} \log \left(\frac{\pi e \xi}{\sigma_{\alpha_k}^2 P(t,\alpha_k)} \right) dk \right| X^\pi_t = x, \alpha_t = i \right],
\end{align}
where the second equality is a direct substitution of the entropy of the Gaussian distribution $\pi^*_k$.
This further gives an expectation of the {\it temporal difference} in the value function from $t$ to $s$: 
\begin{equation*}
    \mathbb{E}\left[\left. \frac{V^*(s, X^{\pi}_s, \alpha_s) - V^*(t,x,i)}{s-t} - \frac{\xi}{s-t} \int_{t}^s \frac{1}{2} \log \left(\frac{\pi e \xi}{\sigma_{\alpha_k}^2 P(t,\alpha_k)} \right) dk \right| X^\pi_t = x, \alpha_t = i \right] = 0.
\end{equation*}
Taking $s \downarrow t$, we can intuitively treat the left-hand-side as the expectation of an instantaneous improvement of the value function from at time $t$, following the investment policy $\pi^*_t$. Such an expectation being zero at all $t \in [0,T]$ is considered ideal because this implies no further improvement is attainable. While this expectation is intractable in the continuous time setting, 
Temporal Difference (TD) Learning allows us to approximate this expectation by its discretized counterpart. 

Consider a discretization of the investment horizon $0 = t_0 < \cdots t_k < \cdots t_{K} = T$, with mesh size equal to $\Delta t := t_{k+1} - t_k$ and $K:= \frac{T}{\Delta t}$. We define the {\it TD loss function} as the mean square sum of the temporal differences measured at $\{t_k\}_{k=0,\cdots,K}$. 
\begin{equation}\label{eq:TD-loss}
\begin{split}
    TD(\theta) = \frac{1}{2} \mathbb{E} \left[\sum_{k=0}^{K-1} \left(\frac{V^\theta(t_{k+1}, X_{t_{k+1}}^\theta, \alpha_{t_{k+1}}) - V^\theta(t_{k}, X_{t_{k}}^\theta, \alpha_{t_{k}})}{\Delta t} - \frac{\xi}{2} \log \left(\frac{\pi e \xi}{\sigma_{\alpha_{t_k}}^2 P^\theta(t_k,\alpha_{t_k})} \right) 
    \right)^2 \Delta t
    \right].
\end{split}
\end{equation}
The above expectation is taken over the filtered probability space of the market and the policy exploration, which can be approximated by simulating trajectories of portfolio values $\{X_{t_k}^\theta\}_{k=1}^K$ and the Markovian regimes $\{\alpha_{t_k}\}_{k=1}^K)$. 
We use Stochastic Gradient Decent (SGD) to minimize the TD loss, which we describe in more detail in Section \ref{sec:updating-scheme-training-algorithm}. 

Here, we remind the readers that the recursive form (Eq. \ref{eq:DPP-recursive-form-of-V*}) only holds for the optimal value function $V^*$, which is supposed to be attained when the market parameters $\theta$ converge to the ``grounding true" values $\theta_{true}$. On the other hand, if we define a process
\begin{equation}
    M^*_t := V^*(t, x, i) + \int_0^t \xi \int_{\mathcal{A}} \pi_k^*(u|\alpha_k) \log \pi_k^*(u|\alpha_k) du dk
\end{equation}
given $X_t^\pi = x \in \mathbb{R}$ and $\alpha_t = i\in \{1,2\}$, then $M^* := \{M^*_t\}_{t\in[0,T]}$ is a martingale according to Eq. \ref{eq:DPP-recursive-form-of-V*}, with the following dynamics
\begin{equation*}
\begin{split}
    d & M_t^* = dV^*(t,x, i) + \left(\xi \int_{\mathcal{A}} \pi_k^*(u|\alpha_k) \log \pi_k^*(u|\alpha_k) du \right) dt \\
    & = \Biggl\{ 
        V_t^*(t,x,i) + V_x^*(t,x,i)r_i x + \sum_{j=1}^2 q_{ij} V^*(t,x,i) - \frac{1}{2} \frac{\rho_i^2 (V_x^*(t,x,i))^2}{V_{xx}^*(t,x,i)} - \frac{\xi}{2} \log \left( \frac{2\pi \xi}{\sigma_i^2 V_xx^*(t,x,i)}\right)
    \Biggr\} dt \\
    & + \Biggl\{ V_x^*(t,x,i) \sigma_i \sqrt{\frac{\xi}{\sigma_i^2 V_{xx}^*(t,x,i)} + \left(\frac{\rho_i V_x^*(t,x,i)}{\sigma_i V_{xx}^*(t,x,i)}\right)^2}
    \Biggr\} dW_t.
\end{split}
\end{equation*}
Based on the TD loss (Eq. \ref{eq:TD-loss}), we have:
\begin{align*}
    & \frac{1}{\Delta t} \sum_{k=0}^{K-1} \left(V^*(t_{k+1}, X_{t_{k+1}}^\pi, \alpha_{t_{k+1}}) - V^*(t_{k}, X_{t_{k}}^\pi, \alpha_{t_{k}}) - \frac{\xi}{2} \int_{t_k}^{t_{k+1}} \log \left(\frac{\pi e \xi}{\sigma_{\alpha_{s}}^2 P(s,\alpha_{s})} \right) ds + \mathcal{O}(\Delta t^2) 
    \right)^2 \\
    & \approx \frac{1}{\Delta t} \langle M^* \rangle_T = \frac{1}{\Delta t} \int_0^T (V_x^*(t, X_t^\pi, \alpha_t))^2 \sigma_{\alpha_t}^2 \left(\frac{\xi}{\sigma_{\alpha_t}^2 V_{xx}^*(t, X_t^\pi, \alpha_t)} + \left(\frac{\rho_{\alpha_t}^2 (V_x^*(t, X_t^\pi, \alpha_t))^2}{\sigma_{\alpha_t}^2 (V_{xx}^*(t, X_t^\pi, \alpha_t))^2}\right)^2\right) dt.  
\end{align*}
This means that training with the TD loss is equivalent to minimizing the quadratic variation of $M^*$, which is not zero and should not be minimized either. This implies that minimizing the TD loss is inadequate for the EMVRS problem. Our findings agree with the arguments in \cite{JMLR:v23:21-0947}, although \cite{wang2020continuous} uses a similar TD loss for the EMV problem. This motivates us to proceed with another method to learn the market parameters. 

\subsection{Orthogonality
Condition (OC) Learning}\label{sec:OC-learning}
For some market parameters $\theta$ and their ``grounding truth" $\theta_{true} := (\sigma_{true, 1}, \sigma_{true, 2}, \rho_{true, 1}, \rho_{true, 2})$, we define
\begin{equation} \label{eq:M_theta_process}
\begin{split}
    M^\theta_t & := V^\theta(t, X_t^\theta, \alpha_t) + \int_0^t \xi \int_{\mathcal{A}} \pi_k^\theta(u|\alpha_k) \log \pi_k^\theta(u|\alpha_k) du dk \\
    & = V^\theta(t,X_t^\theta, \alpha_t) - \frac{\xi}{2} \int_0^t \log\left(\frac{\pi e\xi}{\sigma^2_{\alpha_{k}} P^\theta(k, \alpha_{k})} \right) dk.
\end{split}
\end{equation}
We remind the readers that the portfolio value is derived from the market dynamics, which should use $\theta_{true}$. So, the portfolio value process is rewritten from Eq. \ref{eq:EMVRS-wealth-process-exploratory} to
\begin{align}
    dX_t^\theta & = [r_{\alpha_t} X_t^\theta + \int_\mathcal{A} \rho_{true,\alpha_t}\sigma_{true, \alpha_t} \cdot u \cdot \pi_t^\theta(u|\alpha_t) du] dt + \sqrt{\int_\mathcal{A} \sigma_{true, \alpha_t}^2 u^2 \pi_t^\theta(u|\alpha_t)du} dW_t \\
    \begin{split} \label{eq:EMVRS-wealth-process-with-theta-true}
    & = \Biggl\{ r_{\alpha_t} X_t^\theta + \rho_{true, \alpha_t} \sigma_{true, \alpha_t} \left(-\frac{\rho_{\alpha_t}}{\sigma_{\alpha_t}}[X_t^\theta + (\lambda-z) H^\theta(t,\alpha_t)]\right) \Biggr\} dt \\
    & + \Biggl\{\sigma_{true, \alpha_t} \sqrt{\frac{\rho^2_{\alpha_t}}{\sigma^2_{\alpha_t}} [X_t^\theta + (\lambda-z) H^\theta(t,\alpha_t)]^2 + \frac{\xi}{2\sigma^2_{\alpha_t} P^\theta(t,\alpha_t)}} \Biggr\} dW_t. 
    \end{split} 
\end{align}
Applying It\^o's Formula and after some calculations, we get the drift term of $dM^\theta$
\begin{align*}
  & \Biggl\{ \dot{P}^\theta(t,\alpha_t) + P^\theta(t,\alpha_t) \sigma_{true, \alpha_t}^2 \frac{\rho^2_{\alpha_t}}{\sigma^2_{\alpha_t}} - 2P^\theta(t,\alpha_t) \rho_{true, \alpha_t} \sigma_{true, \alpha_t} \frac{\rho_{\alpha_t}}{\sigma_{\alpha_t}} \\
  & + 2P^\theta(t,\alpha_t) r_{\alpha_t} + \sum_{j=1}^l q_{\alpha_tj} P^\theta(t,j) \Biggr\} \times [X_t + (\lambda-z) H^\theta(t,\alpha_t)]^2 \\
  + & \Biggl\{ \dot{H}^\theta(t,\alpha_t) - r_{\alpha_t} H^\theta(t,\alpha_t) + \frac{1}{P^\theta(t,\alpha_t)} \sum_{j=1}^l q_{\alpha_tj} P^\theta(t,j) [H^\theta(t,j) - H^\theta(t,\alpha_t)] 
  \Biggr\} \\
  & \times 2(\lambda-z) P^\theta(t,\alpha_t) [X_t +(\lambda-z) H^\theta(t,\alpha_t)] \\
  + & \Biggl\{ 
  \dot{C}^\theta(t,\alpha_t) + \sum_{j=1}^l q_{\alpha_t j} \left[P^\theta(t,j)[H^\theta(t,j) -H^\theta(t,\alpha_t)]^2 + C^\theta(t,j) \right] 
  \Biggr\} \times (\lambda-z)^2 \\
  + & \Biggl\{\dot{D}^\theta(t,\alpha_t) + \sum_{j=1}^l q_{\alpha_tj} D^\theta(t,j) - \frac{\lambda}{2} \log \left(\frac{\pi \lambda}{\sigma^2_{\alpha_t} P^\theta(t, \alpha_t)} \right) \Biggr\},
\end{align*}
which is zero if and only if $\theta = \theta_{true}$, given that $P^\theta, H^\theta, C^\theta, D^\theta$ are the solutions to the ODEs in Eq. \ref{eq:MVRS-P(t,i)-ODE}, \ref{eq:MVRS-H(t,i)-ODE}, \ref{eq:EMVRS-C(t,i)} and \ref{eq:EMVRS-D(t,i)}. That is, $M^\theta_t$ is a martingale if and only if $\theta$ coincides with $\theta_{true}$. 

Moreover, we know that any square-integrable martingale $M:=\{M_t\}_{t\in [0,T]}$ has an orthogonality condition 
\begin{equation}
    \mathbb{E} \left[\int_0^T \zeta_t dM_t \right] = 0,
\end{equation}
where the {\it test function} $\zeta:=\{\zeta_t\}_{t\in[0,T]}$ is any $\{\mathcal{F}_t\}$-adapted process that is square-integrable with respect to $M$. According to \cite{JMLR:v23:21-0947}, this is a necessary and sufficient condition to characterize the martingality of $M$, regardless of the choice of the test function. 
Hence, we define the {\it Orthogonality Condition (OC) loss} by taking $\zeta$ as the partial derivative of the value function with respect to the market parameters and discretizing the continuous time. 
\begin{definition}[Orthogonality Condition (OC) loss]
    Let $\theta \equiv (\theta_1, \theta_2, \theta_3, \theta_4) \equiv (\sigma_1, \sigma_2, \rho_1, \rho_2)$ and for $j = 1, \cdots, 4$, 
\begin{equation}\label{eq:orthog-loss}
\begin{split}
    OC&(\theta_j) = \mathbb{E} \left[\sum_{k=0}^{K-1} \frac{\partial V^\theta(t_k, X_{t_k}, \alpha_{t_k})}{\partial \theta_j} \left(M_{t_{k+1}}^\theta - M_{t_k}^\theta \right) \right]  \\
    & = \mathbb{E} \left[\sum_{k=0}^{K-1} \frac{\partial V^\theta(t_k, X_{t_k}, \alpha_{t_k})}{\partial \theta_j} \left(V^\theta(t_{k+1}, X_{t_{k+1}}^\theta, \alpha_{t_{k+1}}) - V^\theta(t_{k}, X_{t_{k}}^\theta, \alpha_{t_{k}}) - \frac{\xi}{2} \log\left(\frac{\pi e\xi}{\sigma^2_{\alpha_{t_k}} P^\theta(t_k, \alpha_{t_k})} \right) \Delta t 
    \right) \right].
\end{split}
\end{equation}
\end{definition}
Again, the above expectation is taken over the filtered probability space of the market and the policy exploration, which can be approximated by simulating trajectories of portfolio values $\{X_{t_k}^\theta\}_{k=1}^K$ and the Markovian regimes $\{\alpha_{t_k}\}_{k=1}^K)$.

\subsection{Updating Scheme and Training Algorithm} \label{sec:updating-scheme-training-algorithm}
We note that the randomness in our problem comes from three sources --- the market dynamics, the Markovian regimes and the policy exploration. Therefore, it is impractical to exactly compute the TD loss (Eq. \ref{eq:TD-loss}) or the OC loss (Eq. \ref{eq:orthog-loss}), as they are defined to be expectations. Instead, we adopt batch-training over a predetermined number of epochs $N_{epochs}$, which is set to be large enough for market parameters to converge. 

For each epoch $n = 1, \cdots, N_{epochs}$, we fix a realized path of Brownian motion $\{W_{t_k}^{(n)}\}_{k=0}^K$, where each $\Delta W_{t_k}^{(n)} := W^{(n)}_{t_{k+1}} - W^{(n)}_{t_k}$ follows a Gaussian distribution with zero mean and variance $\Delta t$. We hereafter use the superscript $(n)$ to denote a realized path of a process that we fix in epoch $n$. This helps us eliminate the randomness of the market dynamics. Then, we simulate a path of market regimes $\{\alpha_{t_k}^{(n)}\}_{k=0}^K$ following the categorical distribution below, given a randomly selected initial regime $\alpha_{t_0}^{(n)} \in \{1,2\}$ and a predefined Markov Chain generator $Q$. 
\begin{equation} \label{eq:simulate-regimes-path}
    \alpha_{t_{k+1}}^{(n)} = 
    \begin{cases}
        1, \text{ with probability } P(\alpha_{t_{k+1}}^{(n)} = 1|\alpha_{t_{k}}^{(n)}) = p_{\alpha_{t_{k}}^{(n)} 1} \\
        2, \text{ with probability } P(\alpha_{t_{k+1}}^{(n)} = 2|\alpha_{t_{k}}^{(n)}) = p_{\alpha_{t_{k}}^{(n)} 2}
    \end{cases}
\end{equation}
Here $(p_{ij})_{i=1,2}^{j=1,2} =: P$ is the transition matrix of the Markov Chain, which is the matrix exponential of the Markov Chain generator, i.e., $P = e^{Q}$. This eliminates the randomness of the Markovian regimes in epoch $n$. Finally, we use the ``grounding true" market parameters $\theta_{true}$ and the current market parameters $\theta^{(n)}$ to compute a path of the portfolio value $\{X^{(n)}_{t_k}\}_{k=0}^K$, given a predetermined initial portfolio value $X_{t_0}^{(n)} = x_0 > 0$
\begin{equation} \label{eq:simulate-portfolio-value-path}
    X_{t_{k+1}}^{(n)} = X_{t_{k}}^{(n)} + \Delta X_{t_k}^{(n)}, \text{ for } k = 1, \cdots, K-1
\end{equation}
where $\Delta X_{t_k}^{(n)}$ is a discretized version of Eq. \ref{eq:EMVRS-wealth-process-with-theta-true} 
\begin{align*}
    \Delta X_{t_k}^{(n)} & = \Biggl\{ r_{\alpha_{t_k}^{(n)}} X_{t_k}^{(n)} + \rho_{true, \alpha_{t_k}^{(n)}} \sigma_{true, \alpha_{t_k}^{(n)}} \left(-\frac{\rho_{\alpha_{t_k}^{(n)}}}{\sigma_{\alpha_{t_k}^{(n)}}}\left[X_{t_k}^{(n)} + (\lambda-z) H^\theta\left(t_k,\alpha_{t_k}^{(n)}\right)\right]\right) \Biggr\} \Delta t \\
    & + \Biggl\{\sigma_{true, \alpha_{t_k}^{(n)}} \sqrt{\frac{\rho^2_{\alpha_{t_k}^{(n)}}}{\sigma^2_{\alpha_{t_k}^{(n)}}} \left[X_{t_k}^{(n)} + (\lambda-z) H^\theta \left(t_k,\alpha_{t_k}^{(n)} \right)\right]^2 + \frac{\xi}{2\sigma^2_{\alpha_{t_k}^{(n)}} P^\theta\left(t_k,\alpha_{t_k}^{(n)}\right)}} \Biggr\} \Delta W_{t_k}^{(n)}. 
\end{align*}
This helps us eliminate the stochasticity of policy exploration within an epoch. We hereby emphasize the hybrid usage of $\theta_{true}$ and $\theta^{(n)}$ when simulating paths of portfolio value. The $\theta_{true}$, i.e., $(\rho_{true}, \sigma_{true})$, in the equation above was directly derived from the stock dynamics in Eq. \ref{eq:problem-formulation-stock-dynamics}. This is supposed to be the ``grounding true" parameters because the stock price is driven by the ``true" market model, despite that we cannot observe the parameters of it. This is the only place where we used $\theta_{true}$ in the algorithm.


Up to this point, we have collected two paths, the portfolio values $\{X^{(n)}_{t_k}\}_{k=0}^K$ and market regimes $\{\alpha_{t_k}^{(n)}\}_{k=0}^K$, with which we can solve for $\{P^{\theta, (n)}(t_k, \alpha_{t_k}^{(n)})\}, \{H^{\theta, (n)}(t_k, \alpha_{t_k}^{(n)})\}, \{C^{\theta, (n)}(t_k, \alpha_{t_k}^{(n)})\}$ and $\{D^{\theta, (n)}(t_k, \alpha_{t_k}^{(n)})\}$ the system of four ODEs in Eq. \ref{eq:MVRS-P(t,i)-ODE}, \ref{eq:MVRS-H(t,i)-ODE}, \ref{eq:EMVRS-C(t,i)} and \ref{eq:EMVRS-D(t,i)}. We can thereby compute the value function $\{V_{t_k}^{\theta, (n)}\}_{k=0}^K$ via Eq. \ref{eq:EMVRS-optimal-value-function-reparam} and the optimal Lagrange multiplier $\lambda^{(n)}$ via a modified version of Eq. \ref{eq:EMVRS-optimal-lambda} 
\begin{equation} \label{eq:EMVRS-optimal-lambda-updated}
    \lambda^{(n)} = \frac{z - P^{(n)}(t_0, \alpha_{t_0}^{(n)}) H^{(n)}(t_0, \alpha_{t_0}^{(n)}) x_0}{P^{(n)}(t_0, \alpha_{t_0}^{(n)}) H^{(n)}(t_0, \alpha_{t_0}^{(n)})^2 + C^{(n)}(t_0, \alpha_{t_0}^{(n)}) - 1} + z.
\end{equation}

We note that it is impractical to compute the TD loss and the OC loss in Eq. \ref{eq:TD-loss} and Eq. \ref{eq:orthog-loss}, due to the multiple sources of randomness. So, we instead compute the {\it realized TD loss} and {\it realized OC loss}, respectively denoted by $TD(\theta^{(n)}; \{X_{t_k}^{(n)}\}, \{\alpha_{t_k}^{(n)}\})$ and $OC(\theta^{(n)}_j; \{X_{t_k}^{(n)}\}, \{\alpha_{t_k}^{(n)}\})$, for $j=1,\cdots,4$.
\begin{align}
\begin{split} \label{eq:TD-loss-realized}
    TD& (\theta^{(n)}; \{X_{t_k}^{(n)}\}, \{\alpha_{t_k}^{(n)}\}) \\
    & = \frac{1}{2} \sum_{k=0}^{K-1} \left(\frac{V^{\theta, (n)}\left(t_{k+1}, X_{t_{k+1}}^{(n)}, \alpha_{t_{k+1}}^{(n)}\right) - V^{\theta, (n)}\left(t_{k}, X_{t_{k}}^{(n)}, \alpha_{t_{k}}^{(n)}\right)}{\Delta t} - \frac{\xi}{2} \log \left(\frac{\pi e \xi}{\sigma_{\alpha_{t_k}^{(n)}}^2 P^{\theta, (n)} \left(t_k,\alpha_{t_k}^{(n)}\right)} \right) 
    \right)^2 \Delta t
\end{split} \\
\begin{split} \label{eq:orthog-loss-realized}
    OC& (\theta^{(n)}_j; \{X_{t_k}^{(n)}\}, \{\alpha_{t_k}^{(n)}\}) 
    = \sum_{k=0}^{K-1} \frac{\partial V^{\theta,(n)}\left(t_k, X_{t_k}^{(n)}, \alpha_{t_k}^{(n)}\right)}{\partial \theta_j} \\
    & \quad \times \Biggl[V^{\theta,(n)}\left(t_{k+1}, X_{t_{k+1}}^{(n)}, \alpha_{t_{k+1}}^{(n)}\right) - V^{\theta,(n)}\left(t_{k}, X_{t_{k}}^{(n)}, \alpha_{t_{k}}^{(n)}\right) - \frac{\xi}{2} \log\left(\frac{\pi e\xi}{\sigma^2_{\alpha_{t_k}^{(n)}} P^{\theta,(n)}\left(t_k, \alpha_{t_k}^{(n)}\right)} \right) \Delta t 
    \Biggr]. 
\end{split}
\end{align}

To update the market parameters, we apply the Stochastic Gradient Decent on the TD loss and the Stochastic Optimization method on the OC loss. For $j=1,\cdots,4$, the updating scheme of $\theta$ with the TD loss is
\begin{equation} \label{eq: TD-update-scheme}
    \theta_j^{(n+1)} \leftarrow \theta_j^{(n)} - \eta_{TD,j} \frac{\partial TD(\theta^{(n)}; \{X_{t_k}^{(n)}\}, \{\alpha_{t_k}^{(n)}\})}{\partial \theta_j},
\end{equation}
and the updating scheme of $\theta$ with the OC loss is
\begin{equation} \label{eq: OC-update-scheme}
    \theta_j^{(n+1)} \leftarrow \theta_j^{(n)} + \eta_{OC,j} OC(\theta^{(n)}_j; \{X_{t_k}^{(n)}\}, \{\alpha_{t_k}^{(n)}\}), 
\end{equation}
where $\eta_{TD, j}, \eta_{OC,j} > 0$ are the learning rates for $\theta_j$ when using the TD loss and the OC loss respectively. In practice, the partial derivative is nontrivial and does not have an explicit form in the EMVRS problem. So, we approximate it via the central difference method
\begin{equation*}
    \frac{\partial TD(\theta^{(n)}; \{X_{t_k}^{(n)}\}, \{\alpha_{t_k}^{(n)}\})}{\partial \theta_j} \approx \frac{TD(\theta^{(n)}_{j+}; \{X_{t_k}^{(n)}\}, \{\alpha_{t_k}^{(n)}\}) - TD(\theta^{(n)}_{j-}; \{X_{t_k}^{(n)}\}, \{\alpha_{t_k}^{(n)}\})}{2 \epsilon_j},
\end{equation*}
where $\epsilon_j >0$ is a small step size and $\theta^{(n)}_{j\pm}$ equals $\theta^{(n)}$ expect the $j$-th entry is replace by $\theta_j \pm \epsilon_j$.

Finally, we summarize our method in the following algorithm. 

\begin{algorithm}[H]
\SetAlgoLined
Initialize the hyperparameters: number of epochs $N_{epochs}$, the investment horizon $T$, the mesh size of the continuous time discretization $\Delta t$, the Markov Chain generator $Q$, the exploration parameter $\xi$, the initial portfolio value $x_0$, the target terminal portfolio value $z$, the learning rates $\eta = (\eta_{1}, \eta_{2}, \eta_{3}, \eta_{4})$, the ``grounding true" market parameters $\theta_{true} = (\theta_{true,1}, \theta_{true,2}, \theta_{true,3}, \theta_{true,4})$ $\equiv (\sigma_{true,1}, \sigma_{true,2}, \rho_{true,3}, \rho_{true,4})$.

Initialize the market parameters $\theta^{(0)} = (\theta_1^{(0)}, \theta_2^{(0)}, \theta_3^{(0)}, \theta_4^{(0)}) \equiv (\sigma_1^{(0)}, \sigma_2^{(0)}, \rho_1^{(0)}, \rho_2^{(0)})$ and the interest rates $(r_1, r_2)$. 

\For{$n=0,\cdots,N_{epochs}-1$}{
Fix a realized path of Brownian motion $\{W_{t_k}^{(n)}\}_{k=0}^K$, with $\Delta W_{t_k}^{(n)} := W^{(n)}_{t_{k+1}} - W^{(n)}_{t_k} \sim N(0, \Delta t)$. 

Simulate $\{X^{(n)}_{t_k}\}_{k=0}^K$ and $\{\alpha_{t_k}^{(n)}\}_{k=0}^K$ via Eq. \ref{eq:simulate-portfolio-value-path} and \ref{eq:simulate-regimes-path}.

Solve the four systems of ODEs in Eq. \ref{eq:MVRS-P(t,i)-ODE}, \ref{eq:MVRS-H(t,i)-ODE}, \ref{eq:EMVRS-C(t,i)} and \ref{eq:EMVRS-D(t,i)} for $\{P^{\theta, (n)}(t_k, \alpha_{t_k}^{(n)})\}, \{H^{\theta, (n)}(t_k, \alpha_{t_k}^{(n)})\}, \{C^{\theta, (n)}(t_k, \alpha_{t_k}^{(n)})\}$ and $\{D^{\theta, (n)}(t_k, \alpha_{t_k}^{(n)})\}$. 

Compute $\{V_{t_k}^{\theta, (n)}\}_{k=0}^K$ via Eq. \ref{eq:EMVRS-optimal-value-function-reparam} and $\lambda^{(n)}$ via Eq. \ref{eq:EMVRS-optimal-lambda-updated}. 

\If{TD learning}{
    Update $\theta_j^{(n)}$ via the updating scheme \ref{eq: TD-update-scheme}, for $j=1,\cdots,4$.
}
\If{OC learning}{
    Update $\theta_j^{(n)}$ via the updating scheme \ref{eq: OC-update-scheme}, for $j=1,\cdots,4$.
}
}
\caption{The Training Algorithm for EMVRS with Simulated Data} 
\label{alg:training-algorithm}
\end{algorithm}

\subsection{Training on Real Data} \label{sec:training-on-real-data}
If we can observe real market data, the source of randomness of our problem reduces to policy exploration, since the observed stock prices $\{S_{t_k}\}_{k=0}^K$ and the interest rates $\{r_{t_k}\}_{k=0}^K$ over a predetermined observation window are fixed. The market regimes $\{\alpha_{t_k}\}_{k=0}^K$, although not directly observable, can be estimated via fitting a Hidden Markov Model and implementing the Viterbi Algorithm (\citet{viterbi1967error}). The Viterbi Algorithm outputs the maximum-a-posterior estimate of the hidden state sequence, which can be taken as a surrogate for the market regimes. 

The training algorithm with real data is the same as Algorithm \ref{alg:training-algorithm}, except we no longer need to sample paths of Brownian motion, nor do we have knowledge of the ``ground true" market parameters. Instead, we generate a path of portfolio values $\{X^{(n)}_{t_k}\}_{k=0}^K$ by first sampling investment action $u^{(n)}_{t_k}$ from the current policy $\pi^\theta_t(\cdot; \alpha^{(n)}_{t_k})$ and then iteratively computing the next portfolio value via 
\begin{equation*}
    X^{(n)}_{t_{k+1}} = u^{(n)}_{t_k} \frac{S_{t_{k+1}}}{S_{t_k}} + (X^{(n)}_{t_k} - u^{(n)}_{t_k}) (1 + r_{t_k} \Delta t)
\end{equation*}
for $k = 0, \cdots, K-1$ and given $X_0^{(n)} = x_0 > 0$. The TD loss and OC loss are defined in the same way as in Eq. \ref{eq:TD-loss} and \ref{eq:orthog-loss}, except the expectations are taken over the filtered probability space of policy exploration only and are conditioned on the observed real market data.


\section{Numerical Results} \label{sec:simulation}

In this section, we demonstrate the performance of our proposed method with both simulated data and real data. 
Two key lessons emerge from the results. 
First, OC learning leads to the parameters' convergence to their corresponding ``grounding truth" in the simulation study, where the ``grounding truths" are given upon initialization. However, TD learning does not guarantee such convergence. Second, EMVRS outperforms EMV on the real data under different investment constraint settings, achieving higher mean terminal portfolio value and relatively lower volatility.

\subsection{Simulation Study: Comparing TD and OC Learning} \label{sec:simulation-compare-TD-OC}

Following Algorithm \ref{alg:training-algorithm}, we first contrast the parameter convergence of TD learning and OC learning with a toy simulation in which only the mean of stock returns are different in the two market regimes and all other market parameters are set equal. 
We consider a one-year investment horizon ($T=1$) with 10 equal-step-size time points throughout the year for portfolio rebalancing ($\Delta t = \nicefrac{1}{10}$). The investor starts with $x_0 = \$1$ and sets a target of $z=\$1.4$ to be achieved by the end of the year. During the investment period, the investor explores the investment strategies with the exploration weight equal to $\xi=0.5$. To mimic the stock market, we consider it has two regimes, ``good" $(i=1)$ and ``bad" $(i=2)$, which are alternating over time on a Markov Chain with generator $Q = \left(\begin{matrix} -1 & 1 \\ 1 & -1 \end{matrix}\right)$. The ``grounding true" mean and volatility of the market are $(\mu_{true, 1}, \sigma_{true,1}) = (0.2, 0.2)$ in the ``good" state and $(\mu_{true, 2}, \sigma_{true, 2}) = (-0.1, 0.2)$ in the ``bad" state. For simplicity, suppose the interest rates are $(r_1, r_2) = (0,0)$, hence the ``grounding true" Sharpe ratios are 
\begin{align*}
    & \rho_{true, 1} = \frac{\mu_1 - r_1}{\sigma_1} = \frac{0.2 - 0}{0.2} = 1\\
    & \rho_{true, 2} = \frac{\mu_2 - r_2}{\sigma_2} = \frac{-0.1 - 0}{0.2} = -0.5
\end{align*}
To parameterize the value functions, we initialize the market parameters at \begin{equation*}
    \theta^{(0)} = (\theta_1^{(0)}, \theta_2^{(0)}, \theta_3^{(0)}, \theta_4^{(0)}) \equiv (\sigma_1^{(0)}, \sigma_2^{(0)}, \rho_1^{(0)}, \rho_2^{(0)}) = (0.1, 0.1, 0.8, -0.3).
\end{equation*}
Finally, to ensure numerical stability during the training process, we set constraints for the range of the volatility $\sigma_1, \sigma_2 \in [0.1, 1]$ and the range of the Sharpe ratios $\rho_1, \rho_2 \in [-2, 2]$. 
We summarize the configuration of this pilot simulation in Table \ref{tab:config-of-pilot-simulation}. We set the learning rates to steadily decrease from their initial values to $1\times 10^{-5}$ over the training epochs.

\begin{table}[h]
    \centering
    \begin{tabular}{c|c}
    \toprule
       Hyperparameters  & Value at Initialization  \\
    \midrule
        $T$ & 1 \\
        $\Delta t$ & 0.1 \\
        $Q$ & $\left(\begin{matrix}
            -1 & 1 \\
            1 & -1 
        \end{matrix}\right)$ \\
        $\xi$ & 0.5 \\
        $x_0$ & 1 \\
        $z$ & 1.4 \\
        $\eta$ for TD loss & $(1\times 10^4, 1\times 10^4, 2\times 10^4, 2\times 10^4)$ \\
        $\eta$ for OC loss & $(1\times 10^4, 1\times 10^4, 1\times 10^3, 1\times 10^3)$ \\ 
        $(r_1, r_2)$ & $(0,0)$ \\
        $\theta_{true}$ & $(\sigma_{true, 1} = 0.2, \sigma_{true,2} = 0.2, \rho_{true,1} = 1.0, \rho_{true,2} = -0.5)$ \\
        $\theta^{(0)}$ & $(0.1, 0.1, 0.8, -0.3)$ \\
        \bottomrule
    \end{tabular}
    \caption{Configuration of the Toy Simulation}
    \label{tab:config-of-pilot-simulation}
\end{table}

The paths of the market parameters over training epochs using the TD loss are given in Figure \ref{fig:TD-loss-convergence-ckpt1}, which indicates that the market parameters did not converge to the ``grounding true" values. While the volatility parameters $(\sigma_1, \sigma_2)$ converged to a wrong level, the Sharpe ratio parameters $(\rho_1, \rho_2)$ diverged, reaching the upper and lower boundaries of the range we set previously. The poor convergence performance of TD learning is as expected. Because minimizing the TD loss is equivalent to minimizing the quadratic variation of $M^\theta$ (Eq. \ref{eq:M_theta_process}), which should not be minimized. 

Figure \ref{fig:OC-loss-convergence-ckpt5} shows the market parameters' updating paths for OC learning. All market parameters converged to their corresponding ``grounding true" level. While we acknowledge that the initial values for the Sharpe ratio parameters $(\rho_1^{(0)}, \rho_2^{(0)}) = (0.8, -0.3)$ are set close to their ``grounding truths" $(\rho_{true, 1}, \rho_{true, 2}) = (1, -0.5)$, we reinitialize them as $(\rho_1^{(0)}, \rho_2^{(0)}) = (0.2, 0.2)$. Training with OC loss again produces converging paths of the market parameters towards the ``grounding true" level, as shown in Figure \ref{fig:OC-loss-convergence-ckpt5-1}. Our results empirically justify our solutions to the EMVRS problem in Theorem \ref{thm:EMVRS-solution}, while simultaneously revealing the superiority of OC learning over TD learning in this problem setting.

\begin{figure}[h]
    \centering
    \includegraphics[width=0.45\textwidth]{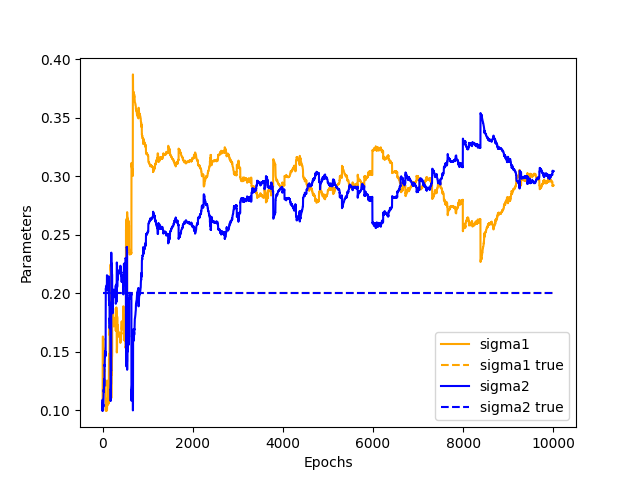}
    \includegraphics[width=0.45\textwidth]{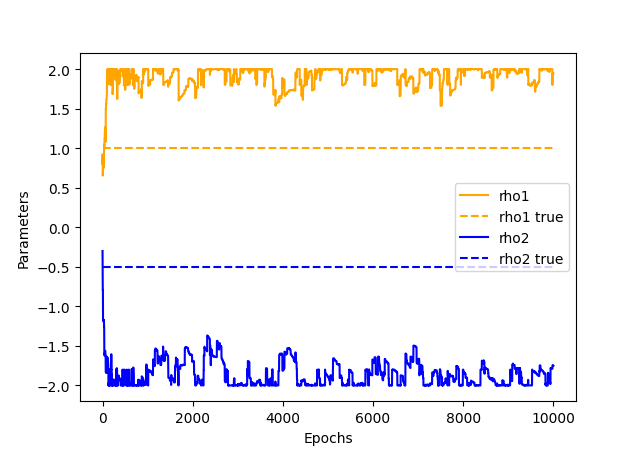}
    \caption{Parameter Convergence of Temporal Difference (TD) Learning. Market parameters are initialized at $\theta^{(0)} = (0.1, 0.1, 0.8, -0.3)$. }
    \label{fig:TD-loss-convergence-ckpt1}
\end{figure}

\begin{figure}[h]
    \centering
    \includegraphics[width=0.45\textwidth]{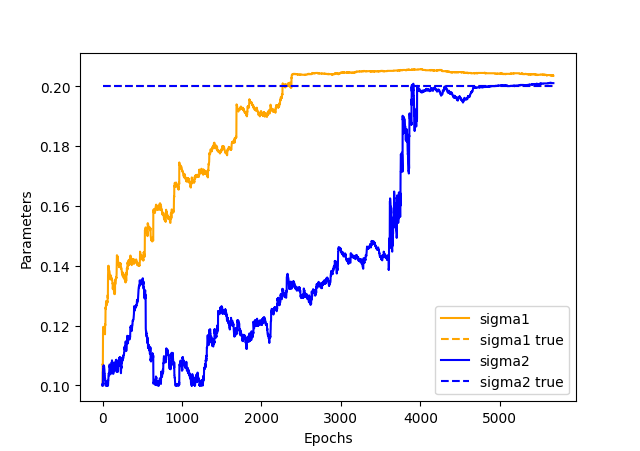}
    \includegraphics[width=0.45\textwidth]{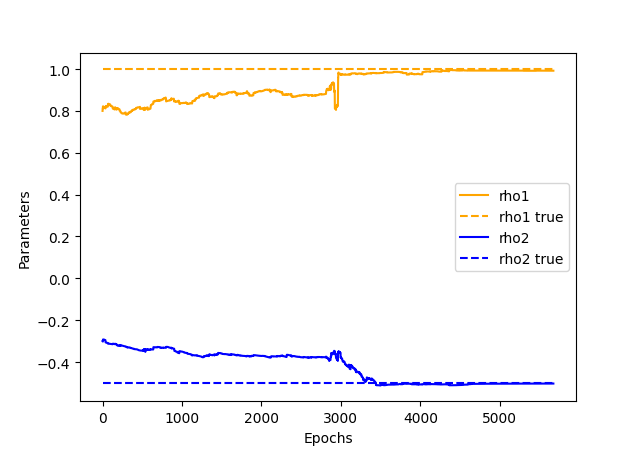}
    \caption{Parameter Convergence of Orthogonality Condition (OC) Learning. Market parameters are initialized at $\theta^{(0)} = (0.1, 0.1, 0.8, -0.3)$.}
    \label{fig:OC-loss-convergence-ckpt5}
\end{figure}

\begin{figure}[h]
    \centering
    \includegraphics[width=0.45\textwidth]{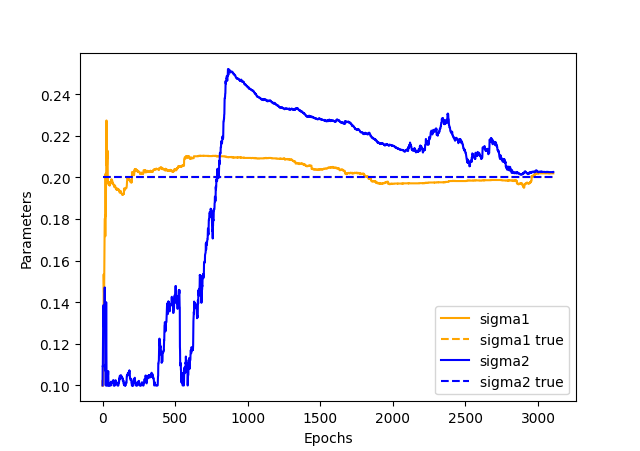}
    \includegraphics[width=0.45\textwidth]{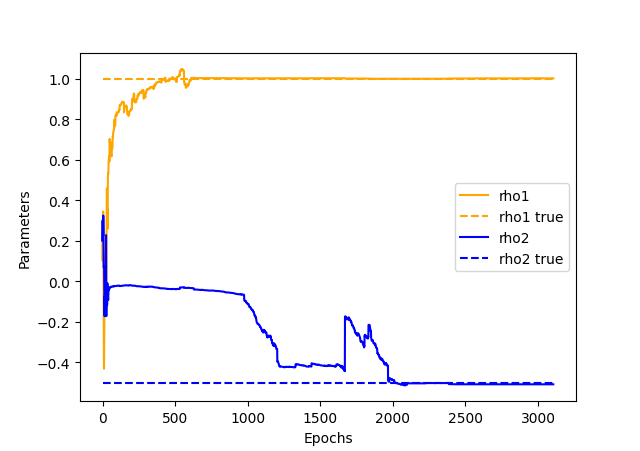}
    \caption{Parameter Convergence Using Orthogonality Condition (OC) Loss. Market parameters are initialized at $\theta^{(0)} = (0.1, 0.1, 0.2, 0.2)$.}
    \label{fig:OC-loss-convergence-ckpt5-1}
\end{figure}

Finally, we examine the robustness of OC learning in a more comprehensive simulation design, where the market regimes are more different from each other.
Table \ref{tab:config-of-full-simulation} summarizes the configuration of this simulation study. We highlight the differences from the previous simulation: we considered the ``grounding true" mean and volatility of the ``good" state market to be $(\mu_{true, 1}, \sigma_{true, 1}) = (0.2, 0.2)$, as the US equity mean return and volatility were recorded at $(18\%, 18.4\%)$ in 1996-2000 according to \citet{bae2014dynamic}. For the ``bad" state, we suppose $(\mu_{true, 2}, \sigma_{true, 2}) = (-0.2, 0.3)$, as the US equity mean return and volatility were recorded at $(-17.4\%, 30.1\%)$ in 2006-2008 according to \citet{bae2014dynamic}. Moreover, we assume the interest rate is lower in the ``good" state and higher in the ``bad" state, hence setting $(r_1, r_2) = (0.01, 0.05)$. So, the ``grounding true" Sharpe ratios are 
\begin{align*}
    & \rho_{true, 1} = \frac{\mu_1 - r_1}{\sigma_1} = \frac{0.2 - 0.01}{0.2} = 0.95, \\
    & \rho_{true, 2} = \frac{\mu_2 - r_2}{\sigma_2} = \frac{-0.2 - 0.05}{0.3} = -0.8\dot{3}.
\end{align*}
We set the market parameters at initialization as
\begin{equation*}
    \theta^{(0)} = (\theta_1^{(0)}, \theta_2^{(0)}, \theta_3^{(0)}, \theta_4^{(0)}) \equiv (\sigma_1^{(0)}, \sigma_2^{(0)}, \rho_1^{(0)}, \rho_2^{(0)}) = (0.1, 0.1, 0.5, -0.5)
\end{equation*}

\begin{table}[H]
    \centering
    \begin{tabular}{c|c}
    \toprule
       Hyperparameters  & Value at Initialization  \\
    \midrule
        $T$ & 1 \\
        $\Delta t$ & 0.1 \\
        $Q$ & $\left(\begin{matrix}
            -1 & 1 \\
            1 & -1 
        \end{matrix}\right)$ \\
        $\xi$ & 0.5 \\
        $x_0$ & 1 \\
        $z$ & 1.4 \\
        $\eta$ for OC loss & $(1\times 10^4, 1\times 10^4, 1\times 10^4, 1\times 10^4)$ \\ 
        $(r_1, r_2)$ & $(0.01, 0.05)$ \\
        $\theta_{true}$ & $(\sigma_{true, 1} = 0.2, \sigma_{true,2} = 0.3, \rho_{true,1} = 0.95, \rho_{true,2} = -0.833)$ \\
        $\theta^{(0)}$ & (0.1, 0.1, 0.5, -0.5) \\
        \bottomrule
    \end{tabular}
    \caption{Configuration of the More Comprehensive Simulation}
    \label{tab:config-of-full-simulation}
\end{table}

The convergence of the market parameters is illustrated in Figure \ref{fig:OC-loss-convergence-ckpt6-3}, which shows that all market parameters converged to a reasonably close range of their corresponding ``grounding truths." We emphasize that even though the market parameters of each regime are diverse in more dimensions, including the stock return volatility and interest rate, the time to convergence wasn't significantly increased compared to that of the toy simulation. 
\begin{figure}[H]
    \centering
    \includegraphics[width=0.45\textwidth]{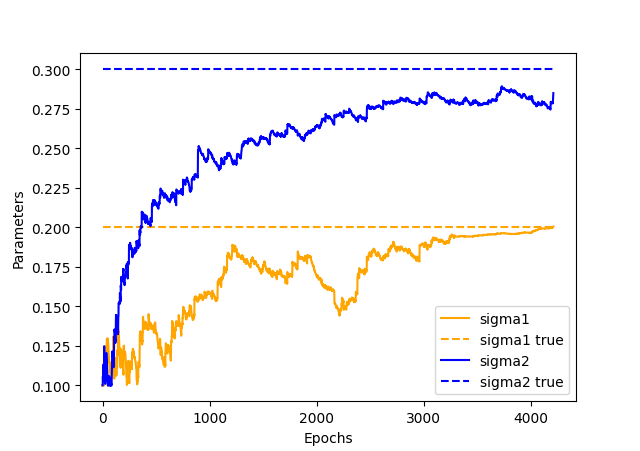}
    \includegraphics[width=0.45\textwidth]{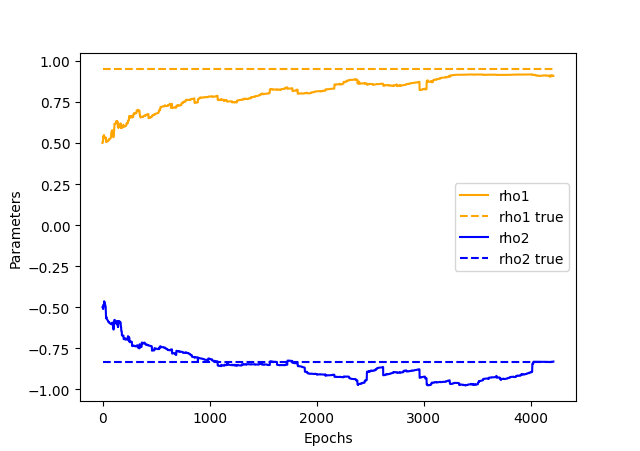}
    \caption{Parameter Convergence of Orthogonality Condition (OC) Learning. Market parameters are initialized at $\theta^{(0)} = (0.1, 0.1, 0.5, -0.5)$.}
    \label{fig:OC-loss-convergence-ckpt6-3}
\end{figure}

\subsection{Performance on Real Data: Comparing EMVRS and EMV}

Having witnessed the parameters' convergence in simulation studies, we evaluate the investment performance of EMVRS on real data, by comparing the mean and volatility of the terminal portfolio values. 
For simplicity, we consider an investor managing a market portfolio and financing under the riskfree interest rate. Hence, we collect from Yahoo Finance daily data of the S\&P500 (ticker: ``\^{}GSPC") market index as the market portfolio, and benchmark the 3-month US Treasury Bill (3mTbill) rate (ticker: ``\^{}IRX") as the riskfree interest rate. 
We consider monthly portfolio rebalancing and set a 10-year rolling window starting from January 1st, 2006, with a one-month step size. So, the first window covers from January 1st, 2006 to December 31st, 2015, the second window covers from February 1st, 2006 to January 31st, 2016, and so on. We keep rolling for two years, hence the last window covers from December 1st, 2007 to November 30th, 2017, resulting in 24 10-year windows. 
Since the rolling windows cover the subprime mortgage crisis and the recovery since then, we broadly categorize the market into two regimes --- a bullish regime and a bearish regime. 
On each rolling window, we apply the Viterbi algorithm \citep{viterbi1967error} 
to identify the market regimes, then train the model using the parameters from the last time frame as the initial points. 
Figure \ref{fig:realdata-training} shows the regime-labeled time series of the first rolling window, where state 0 stands for bullish and 1 stands for bearish. 

\begin{figure}[h]
    \centering
    \includegraphics[width=0.49\linewidth]{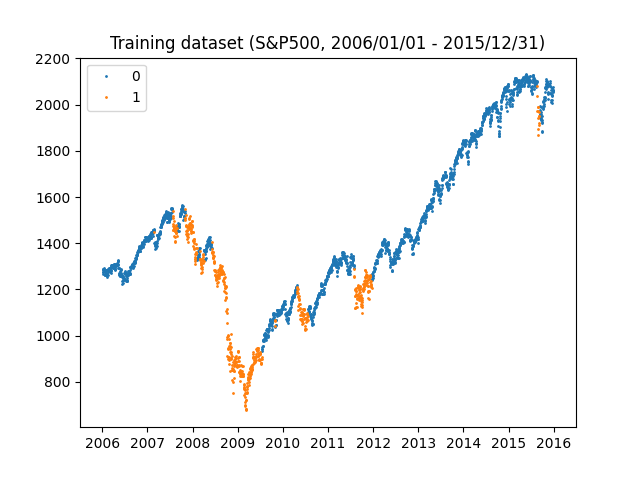}
    \includegraphics[width=0.49\linewidth]{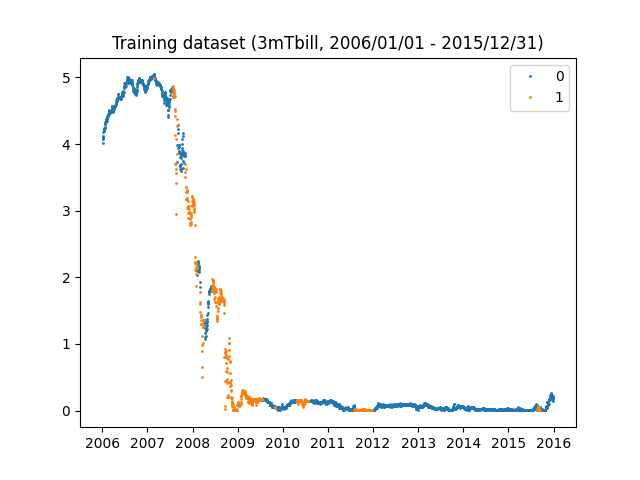}
    \caption{Regime-Identified Financial Time Series for Training (left: S\&P500, right: 3mTbill)}
    \label{fig:realdata-training}
\end{figure}

Due to monthly rebalancing, we mark 120 time points on each 10-year window. 
We suppose the investor starts with an initial wealth $x_0 = \$1$ and targets the terminal portfolio value at $z=\$1.4$ by the end of ten years, which equates to a target annual return of approximately $3.422\%$. In order to maintain numerical stability throughout training, we standardize the S\&P500 index and set an {\it action constraint} at 3, i.e., after sampling an action from the investment policy, we restrict the action value within $[-3x_0, 3x_0] = [-3,3]$.\footnote{This in practice means the investor's leverage ratio is at most 3. The investor can borrow no more than 3 times of their initial wealth $x_0$ to invest in the stock index, or short sell no greater than 3$x_0$ worth of stock index.} We summarize the training configuration and provide the initial parameters for the first rolling window in Table \ref{tab:config-of-realdata}.

\begin{table}[h]
    \centering
    \begin{tabular}{c|c}
    \toprule
       Hyperparameters  & Value at Initialization  \\
    \midrule
        $T$ & 10 \\
        $\Delta t$ & $\frac{1}{12}$ \\
        $\xi$ & 0.5 \\
        $x_0$ & 1 \\
        $z$ & 1.4 \\
        action constraint & 3 \\
        $\eta$ for TD loss & $(1\times 10^3, 1\times 10^3, 1\times 10^3, 1\times 10^3)$ \\ 
        $\eta$ for OC loss & $(1\times 10^3, 1\times 10^3, 1\times 10^3, 1\times 10^3)$ \\ 
        $\theta^{(0)}$ & (0.2, 0.2, 1.0, 1.0) \\
        \bottomrule
    \end{tabular}
    \caption{Configuration of the Real Data Study}
    \label{tab:config-of-realdata}
\end{table}

For comparison, we simultaneously train the EMVRS and EMV via both OC learning and TD learning and backtest the investment performance on the 24 10-year rolling windows. 
We examine the trained models on various investment settings, including action constraints equal to 1, 1.5, 2, 3, and whether short selling is allowed. For example, an action constraint equal to 1 with short selling means the action can take values in $[-x_0, x_0] = [-1,1]$, whereas action constraint equals to 1 without short selling means the action can take values in $[0, x_0] = [0,1]$. We remind the readers that the actions are defined as the money values invested in the market portfolio (i.e., S\&P500), hence an action taking a negative value is a short position. Under each investment setting, we independently trade 5 times over each of the 24 10-year rolling windows, resulting in 100 portfolio value trajectories. We then compute the mean and standard deviation of the last entry of these 100 trajectories as the estimates of the mean and volatility of the terminal portfolio values, respectively, which allows us to compute the annualized portfolio returns and volatilities. 
Since the risk free interest rate is not constant in our problem, we computed the average terminal portfolio value if the investor holds the risk free asset over the 24 10-year rolling windows, which is around $\$1.0747$. This yields the annualized risk free rate at around $0.723\%$. Thereby, we can compute the Sharpe Ratios and we summarize the investment performances on the 24 10-year rolling windows in Table \ref{tab:performance-realdata}. 

Under OC learning, EMVRS outperforms EMV in all of the investment settings we considered. Although we set a low investment target of portfolio return at $3.442\%$, the mean returns of the EMVRS all significantly exceed the predetermined target. Yet, the mean returns of EMV are fairly closed to the target. EMV even fails to reach the target mean return under one investment setting, when action constraint =1 and short selling is not allowed. While both the mean and volatility of the portfolio returns are generally higher for EMVRS than those for EMV, the Sharpe Ratios of EMVRS are significantly higher than those of the EMV.

Under TD learning, EMVRS achieves lower mean portfolio returns but significantly inflated volatilities, which yields the lowest Sharpe Ratios under all investment settings. The performance of EMV is comparable to itself if trained with OC learning, with mean portfolio returns closed to the target. The Sharpe Ratios are not as good as those of EMVRS under OC learning. The poor performance aligns with our previous criticism on TD learning in our problem context. 


Moreover, we also notice that under the same short selling setting, both the mean and volatility of the portfolio returns increase as the action constraint is gradually relaxed, resulting in gradually decreasing Sharpe Ratios. Except for EMVRS with OC learning, the performance on all other model settings demonstrates higher Sharpe Ratios when short selling is forbidden. 
These empirical observations reminds us of being cautious of using high leverage and short position on our investment.

\begin{table}[h]
    \centering
    \begin{tabular}{|l|c c|c c c|c c c|}
    \toprule
    Training & 
    \multicolumn{2}{c|}{Investment Setting} &
    \multicolumn{3}{c|}{EMVRS} &
    \multicolumn{3}{c|}{EMV} \\
    & AC & SS & Mean  & Volatility & SR  & Mean  & Volatility & SR \\
    \midrule
    \multirow{8}{*}{OC 
    Learning} & 
    1 & \cmark & 12.177\% & 1.932\% & 5.9269 & 3.507\% & 1.673\% & 1.6650 \\
    & 1 & \xmark & 12.177\% & 1.932\% & 5.9269 & 3.690\% & 0.857\% & 3.4594 \\
    & 1.5 & \cmark & 15.512\% & 2.890\% & 5.1185 & 3.331\% & 2.419\% & 1.0787 \\
    & 1.5 & \xmark & 15.512\% & 2.890\% & 5.1185 & 3.927\% & 1.246\% & 2.5709 \\
    & 2 & \cmark & 18.156\% & 3.849\% & 4.5318 & 3.762\% & 3.226\% & 0.9420 \\
    & 2 & \xmark & 18.156\% & 3.849\% & 4.5318 & 3.877\% & 1.736\% & 1.8173 \\
    & 3 & \cmark & 22.244\% & 5.771\% & 3.7300 & 3.624\% & 4.917\% & 0.5901 \\
    & 3 & \xmark & 22.244\% & 5.771\% & 3.7300 & 4.066\% & 2.176\% & 1.5404 \\
    \midrule 
    \multirow{8}{*}{TD Learning} & 
    1 & \cmark & 7.319\% & 18.439\% & 0.3577 & 3.426\% & 1.512\% & 1.7800 \\
    & 1 & \xmark & 8.694\% & 14.322\% & 0.5566 & 3.690\% & 0.844\% & 3.5119 \\
    & 1.5 & \cmark & 8.342\% & 21.519\% & 0.3541 & 3.406\% & 2.511\% & 1.0693 \\
    & 1.5 & \xmark & 10.002\% & 17.320\% & 0.5358 & 3.937\% & 1.246\% & 2.5827 \\
    & 2 & \cmark & 8.759\% & 23.309\% & 0.3448 & 3.590\% & 3.365\% & 0.8521 \\
    & 2 & \xmark & 10.720\% & 18.167\% & 0.5503 & 3.887\% & 1.727\% & 1.8309 \\
    & 3 & \cmark & 9.095\% & 25.080\% & 0.3338 & 3.117\% & 4.345\% & 0.5510 \\
    & 3 & \xmark & 11.394\% & 19.223\% & 0.5551 & 4.081\% & 2.188\% & 1.5344 \\
    \bottomrule
    \end{tabular}
    \caption{Annualized Mean and Volatility of Portfolio Returns and Sharpe Ratios on the 24 10-year Windows. (AC: action constraint; SS: short selling; SR: Sharpe Ratio; target annual portfolio return: $3.422\%$)}
    \label{tab:performance-realdata}
\end{table}

\section{Conclusion} \label{sec:conclusion}

Inspired by the past works on the regime-switching MV problem and the recent RL advances in stochastic control, we extended the classical MV portfolio optimization problem and formulated the EMVRS problem. We recognize that the stochastic control approach to the MV problems is optimal given that the market parameters are fully known. However, in practice the market parameters cannot be directly observed, which drives us to explore within the control space. Here, the RL framework offers an informed guidance of exploration within the control space, which is a Gaussian exploration with entropy regularization. Therefore, a combination of the RL framework and the stochastic control solution results in a collaborative exploration and exploitation of the optimal control policy. In this work, we derived an analytical solution to the EMVRS problem, which fills the gap between the Regime-Switching application in the stochastic control literature and the RL framework for the non-Regime-Switching MV problem. 

Furthermore, we fully leveraged the optimal stochastic control solution to the EMVRS problem, by adopting the functional form of the induced optimal value function and reparameterizing the value function with the market parameters. This differs from the approach in \cite{wang2020continuous}, which reparameterizes the value function with a new set of parameters. Incorporating the induced optimal value function provides an informed and theory-backed starting point for the RL training algorithm.
Using the market parameters as the parameterization of the RL algorithm reduces the dimension of the parameter space and produces more meaningful outputs when the parameter estimates have converged. 



By adopting OC learning, our RL algorithm enables EMVRS to successfully recover the hidden market parameters in our simulation studies. The improved investment performance of EMVRS on the real market data also supports our proposed methods in the practical domain. 

\bigskip
We conclude with a potential direction for future work. We mentioned three major sources of randomness in the EMVRS formulation --- the market dynamics, the exploration of the investment strategy and the market regimes. While the market regime is modelled on a Markov Chain, the dynamics of the market (Eq. \ref{eq:problem-formulation-stock-dynamics} and \ref{eq:problem-formulation-bond-dynamics}) and the portfolio value process (Eq. \ref{eq:problem-formulation-wealth-process-classic}) are driven by the same Brownian Motion. This could potentially be generalized to different Brownian Motions with some covariance, separating the market stochasticity from exploration randomness; See \citet{dai2023learning} for the stochasticity separation with different Brownian Motions (although this work finds an equilibrium strategy instead of a pre-commitment policy).

\newpage 

\printbibliography[title=Bibliography]

\newpage
\appendix

\section{Proofs and Derivations}

\subsection{Derivation of Eq. \ref{eq:EMVRS-wealth-process-exploratory}}
We herein derive Eq. \ref{eq:EMVRS-wealth-process-exploratory} from Eq.  \ref{eq:problem-formulation-wealth-process-classic}, inspired by the exploratory formulation arguments in \citep{wang2020reinforcement}.
Recall the original portfolio value dynamic without the exploratory extension (Eq. \ref{eq:problem-formulation-wealth-process-classic})
\begin{equation*}
    d X^u_t = \left[r(t, \alpha_t) X^u_t + \rho(t, \alpha_t) \sigma(t, \alpha_t) u_t \right] dt + \sigma(t, \alpha_t) u_t dW_t.
\end{equation*}
With the exploratory extension, we know that each $u_t$ at time $t$ is sampled from the policy distribution $\pi_t(\cdot|\alpha_t)$ given the regime $\alpha_t$. Consider simulating a path of portfolio values $\{X_t\}_{t\in[0,T]}$ and a path of investment controls $\{u_t\}_{t\in[0,T]}$. Then for any time $t \in [0,T]$ and a small time interval $\Delta t$, 
\begin{align*}
    \Delta X_t = X_{t+\Delta t} - X_t \approx \left[r(t, \alpha_t) X_t + \rho(t, \alpha_t) \sigma(t, \alpha_t) u_t \right] \Delta t + \sigma(t, \alpha_t) u_t (W_{t+\Delta t} - W_t).
\end{align*}
If we repetitively simulate $N$ paths (denoting the $i$-th path with a superscript $i$) and compute
\begin{align*}
    \frac{1}{N} \sum_{i=1}^N \Delta & X_t^i \approx \frac{1}{N} \sum_{i=1}^N \left[r(t, \alpha_t) X_t^i + \rho(t, \alpha_t) \sigma(t, \alpha_t) u_t^i \right] \Delta t + \frac{1}{N} \sum_{i=1}^N \sigma(t, \alpha_t) u^i_t (W^i_{t+\Delta t} - W^i_t) \\
    & \overset{a.s.}{\to} \mathbb{E} \left[\int_\mathcal{U} \left[r(t, \alpha_t) X_t^\pi + \rho(t, \alpha_t) \sigma(t, \alpha_t) u \right] \pi(u|\alpha_t) du \Delta t \right] 
    + \mathbb{E}\left[\int_\mathcal{U} \sigma(t, \alpha_t) u \pi(u|\alpha_t) du \right] \mathbb{E}\left[W^i_{t+\Delta t} - W^i_t\right] \\ 
    & = \mathbb{E}\left[r(t, \alpha_t) X_t^\pi + \int_\mathcal{U} \rho(t, \alpha_t) \sigma(t, \alpha_t) u  \pi(u|\alpha_t) du \right] \Delta t\\
    \frac{1}{N} \sum_{i=1}^N & \left(\Delta X_t^i\right)^2 \approx \frac{1}{N} \sum_{i=1}^N \left(\sigma(t, \alpha_t) u^i_t \right)^2 \Delta t 
    \overset{a.s.}{\to} \mathbb{E} \left[\int_\mathcal{U} \sigma^2(t, \alpha_t) u^2 \pi(u|\alpha_t)  du \right] \Delta t,
\end{align*}
where the almost sure convergence, $\overset{a.s.}{\rightarrow}$, follows from the Law of Large Numbers. Moreover, since the sample paths of portfolio values are randomly and independently simulated from its distribution $X^\pi$, the Law of Large Numbers again implies
\begin{align*}
    \frac{1}{N} \sum_{i=1}^N \Delta X_t^i & \overset{a.s.}{\to} \mathbb{E} \left[\Delta X^\pi_t \right], \\
    \frac{1}{N} \sum_{i=1}^N \left(\Delta X_t^i\right)^2 & \overset{a.s.}{\to}  \mathbb{E} \left[\left(\Delta X_t^\pi \right)^2\right]. 
\end{align*}
This motivates us to formulate the exploratory version of the portfolio value process
\begin{align*} 
    \begin{split}
    d X^\pi_t & = \left[r(t, \alpha_t) X^\pi_t + \int_{\mathcal{A}} \rho(t, \alpha_t) \sigma(t, \alpha_t) \cdot u \cdot \pi_t(u|\alpha_t) du \right] dt + \left(\sqrt{\int_\mathcal{A} \sigma^2(t, \alpha_t) u^2 \cdot \pi_t(u|\alpha_t)  du} \right) dW_t.
    \end{split}
\end{align*}

\subsection{Proof of Theorem \ref{thm:EMVRS-solution}}

We first state and prove the following lemma, which will be used in the proof of Theorem \ref{thm:EMVRS-solution}. 

\begin{lemma} \label{lm:solving-general-ODE-for-C(t,i)-D(t,i)}
For time $t \in [0,T]$ and Markovian regime $i \in \{1, \cdots, l\}$, let $f(t, i): [0,T] \mapsto \mathbb{R}$ and $G(t, i): [0,T] \mapsto \mathbb{R}$ be two $\{\mathcal{F}_t\}_{t\in[0,T]}$-measurable functions such that, for any regime $i \in \{1, \cdots, l\}$,
\begin{itemize}
    \item[(i)] $f(\cdot, i)$ is continuous in time $t \in [0,T]$;
    \item[(ii)] $G(\cdot, i) \in \mathbb{C}^{1}([0,T])$, i.e., $G(t,i)$ is continuously differentiable with respect to $t \in [0,T]$.
\end{itemize} 
We further suppose $f(t,i)$ and $G(t, i)$ also satisfy the following ODE with a terminal condition: 
\begin{equation*}
\begin{cases}
    \dot{G}(t,i) = - f(t,i) - \sum_{j=1}^l q_{ij} G(t,j) \\
    G(T, i) = 0
\end{cases}
\end{equation*}
where $q_{ij}$ is the $(i,j)$-entry of the Markov Chain generator matrix for the regime. 
Then, for $0 < t' \leq t \leq T$, it can be solved by $G(t, \alpha_{t'}) = \mathbb{E} \left[\left. \int_t^T f(s, \alpha_s) ds  \right| \alpha_{t'}\right]$, where the Markovian regime $\{\alpha_t\}_{t \in [0,T]} \in \{1, \cdots, l\}$ is defined in Section \ref{sec:problem-formulation}.
\end{lemma}
\begin{proof}
    For $0 < t' \leq t \leq T$, let $G(t, \alpha_{t'}) = \mathbb{E} \left[\left. \int_t^T f(s, \alpha_s) ds  \right| \alpha_{t'}\right]$. We verify that it satisfies the ODE system in \ref{lm:solving-general-ODE-for-C(t,i)-D(t,i)}. First, the boundary condition is trivially satisfied, regardless of what regime $\alpha_{t'}$ is. To verify the first equation of the ODE, we notice that $\dot{G}$ can be expressed as a limit:
    $$\dot{G}(t, \alpha_t) = \lim_{\Delta t \to 0} \frac{1}{\Delta t} \left[G(t+\Delta t, \alpha_{t+\Delta t}) - G(t,\alpha_t)\right],$$
    where $\Delta t$ is a small time interval. Moreover, on the right hand side, we have:  
    \begin{align*}
    G & (t+\Delta t, \alpha_{t+\Delta t}) - G(t,\alpha_t)  \\
    & = G(t+\Delta t, \alpha_{t+\Delta t}) - G(t+\Delta t, \alpha_t) + G(t+\Delta t, \alpha_t) - G(t,\alpha_t) \\
    & = G(t+\Delta t, \alpha_{t+\Delta t}) 
    - \mathbb{E} \left[ \left.\int_{t+\Delta t}^{T} f(s, \alpha_s) ds \right| \alpha_t \right]  \\
    & + \mathbb{E} \left[\left. \int_{t+\Delta t}^{T} f(s, \alpha_s) ds \right| \alpha_t \right] 
    - \mathbb{E} \left[\left. \int_{t}^{T} f(s, \alpha_s) ds \right| \alpha_t\right] \\
    & = G(t+\Delta t, \alpha(t+\Delta t)) 
    - \mathbb{E} \left[ \left.\int_{t+\Delta t}^{T} f(s, \alpha_s) ds \right| \alpha(t) \right]  
    + \mathbb{E} \left[\left. - \int_{t}^{t+\Delta t} f(s, \alpha_s) ds \right| \alpha_t \right].
    \end{align*}
    Applying the tower property of conditional expectations yields:
    \begin{align*}
        \mathbb{E} \left[ \left.\int_{t+\Delta t}^{T} f(s, \alpha_s) ds \right| \alpha_t \right]
        & = \mathbb{E} \left[ \mathbb{E} \left[ \left. \left. \int_{t+\Delta t}^{T} f(s, \alpha_s) ds \right| \alpha_{t+\Delta t} \right] \right| \alpha_t\right]  \\
        & = \sum_{j=1}^l p_{\alpha_t j}(\Delta t) \mathbb{E} \left[ \left. \int_{t+\Delta t}^{T} f(s, \alpha_s) ds \right| \alpha_{t+\Delta t}=j \right] \\
        & = \sum_{j=1}^l p_{\alpha_t j}(\Delta t) G(t+\Delta t,j), 
    \end{align*}
    where $p_{\alpha_t j} = \mathbb{P}(\alpha_{t+\Delta t} = j | \alpha_t)$. 
    Therefore, 
    \begin{align*}
    G & (t+\Delta t, \alpha_{t+\Delta t}) - G(t,\alpha_t)  \\
    & = 
    - \sum_{j=1}^l p_{\alpha_t j}(\Delta t) \left[ G(t+\Delta t,j) - G(t+\Delta t, \alpha_{t+\Delta t}) \right]
    - \mathbb{E} \left[\left. \int_{t}^{t+\Delta t} f(s, \alpha_s) ds \right| \alpha_t \right],
    \end{align*}
    as $\sum_{j=1}^l p_{\alpha_t j}(\Delta t) = 1$. 
    Finally, divide both sides by $\Delta t$ and take the limit of $\Delta t \to 0+$. The left hand side becomes $\dot{G}$. On the right hand side, for $\alpha_t \neq j$, 
    \begin{align*}
        \lim_{\Delta t \to 0+} \frac{p_{\alpha_t j}(\Delta t)}{\Delta t} = q_{\alpha_t j}
    \end{align*}
    and when $\alpha_t=j$, $\lim_{\Delta t \to 0+} G(t + \Delta t, j) - C(t+\Delta t, \alpha_{t+\Delta t}) = 0$. This yields: 
    \begin{align*}
        \dot{G}(t,\alpha_t) & = - f(t,\alpha_t) - \sum_{j=1}^l q_{\alpha_t j} [G(t,j) - G(t,\alpha_t)] \\ 
        & = - f(t,\alpha_t) - \sum_{j=1}^l q_{\alpha_t j} G(t,j), 
    \end{align*}
    where the last equation is given by $\sum_{j=1}^l q_{\alpha_t j} = 0$. 
\end{proof}

Now, we are ready to prove Theorem \ref{thm:EMVRS-solution}.

\begin{proof}
The optimal policy distribution Eq. \ref{eq:EMVRS-optimal-policy-distribution} can be easily derived from the optimal value function Eq. \ref{eq:EMVRS-optimal-value-function}, following Eq. \ref{eq:EMVRS-optimal-policy-distribution-raw}. Moreover, we notice that the optimal value function can be rewritten as a quadratic function of $(\lambda -z)$:
\begin{equation*}
    V^*(t,x,i) = [P(t,i) H(t,i)^2 + C(t,i) -1](\lambda - z)^2 + 2[P(t,i)H(0,i) x - z](\lambda-z) + P(t,i) x^2 + D(t,i) - z^2.
\end{equation*}
This yields the optimal Lagrange multiplier at initialization $t=0$, which is the minimizer of $V^*(0, x_0, i_0)$, for some initial wealth $x_0 > 0$ and initial regime $i_0 \in \{1, \cdots, l\}$.  

The remainder of the proof is to verify that Eq. \ref{eq:EMVRS-optimal-value-function} solves the reduced HJB equation (Eq. \ref{eq:EMVRS-HJB}). We note that the partial derivatives of $V^*$ are 
\begin{align*}
    V^*_t(t,x,i) & = \dot{P}(t,i)[x + (\lambda-z)H(t,i)]^2 + 2(\lambda-z) P(t,i) [x + (\lambda-z)H(t,i)] \cdot \dot{H}(t,i) \\
    & + (\lambda-z)^2 \dot{C}(t,i) + \dot{D}(t,i), \\
    V^*_x(t,x,i) & = 2P(t,i) [x + (\lambda-z)H(t,i)], \\
    V^*_{xx}(t,x,i) & = 2P(t,i).
\end{align*}
Substituting these into the left hand side of Eq. \ref{eq:EMVRS-HJB} yields
\begin{equation}
\begin{split}
    & \dot{P}(t,i)[x + (\lambda-z)H(t,i)]^2 + 2(\lambda-z) P(t,i) [x + (\lambda-z)H(t,i)] \cdot \dot{H}(t,i) + (\lambda-z)^2 \dot{C}(t,i) + \dot{D}(t,i) \\
     & + \sum_{j=1}^l q_{ij} \left\{P(t,j)[x + (\lambda-z) H(t,j)]^2 + (\lambda-z)^2 C(t,j) + D(t,i) - \lambda^2 \right\} \\
     & + 2 r(t,i)x P(t,i) [x + (\lambda-z)H(t,i)] \\
     & - \rho^2(t,i) P(t,i)[x + (\lambda-z) H(t,i)]^2 - \frac{\xi}{2} \log \left(\frac{ \pi \xi}{\sigma^2(t,i) P(t,i)}\right).
\end{split}
\end{equation}
Adding 3 lines has no effect, but helps isolating $P, H, C, D$: 
\begin{equation}
\begin{split}
    & \dot{P}(t,i)[x + (\lambda-z)H(t,i)]^2 + 2(\lambda-z) P(t,i) [x + (\lambda-z)H(t,i)] \cdot \dot{H}(t,i) + (\lambda-z)^2 \dot{C}(t,i) + \dot{D}(t,i) \\
     & + \sum_{j=1}^l q_{ij} \left\{P(t,j)[x + (\lambda-z) H(t,j)]^2 + (\lambda-z)^2 C(t,j) + D(t,j)\right\} \\ 
     & {  + \sum_{j=1}^l q_{ij} P(t,j) [x + (\lambda-z) H(t,i)]^2 
     - \sum_{j=1}^l q_{ij} P(t,j) [x + (\lambda-z) H(t,i)]^2 } \\
     & {  + \sum_{j=1}^l q_{ij} P(t,j) \left[ 2(\lambda-z)^2 H(t,i)(H(t,j) - H(t,i)) - 2(\lambda-z)^2 H(t,i)(H(t,j) - H(t,i)) \right]} \\
     & + 2 r(t,i)x P(t,i) [x + (\lambda-z)H(t,i)] \\
     & {  + 2 r(t,i) (\lambda -z) P(t,i) [x + (\lambda-z)H(t,i)] H(t,i) - 2 r(t,i) (\lambda -z) P(t,i) [x + (\lambda-z)H(t,i)] H(t,i)} \\
     & - \rho^2(t,i) P(t,i)[x + (\lambda-z) H(t,i)]^2 - \frac{\xi}{2} \log \left(\frac{\pi \xi}{\sigma^2(t,i) P(t,i)}\right). 
\end{split}
\end{equation}
Rearranging and grouping gives 
\begin{equation}
\begin{split}
    & \left\{\dot{P}(t,i) - (\rho^2(t,i) - 2r(t,i)) P(t,i) + \sum_{j=1}^l q_{ij}P(t,j) \right\} \cdot [x + (\lambda-z) H(t,i)]^2 \\
    + & \left\{\dot{H}(t,i) - r(t,i) H(t,i) + \frac{1}{P(t,i)} \sum_{j=1}^l q_{ij} P(t,j) (H(t,j) - H(t,i))  \right\} \\
    & \cdot 2(\lambda-z)P(t,i)[x + (\lambda-z) H(t,i)] \\ 
    + & \left\{ \dot{C}(t,i) + \sum_{j=1}^l q_{ij} \left[P(t,j) (H(t,j) - H(t,i))^2 + C(t,j)\right]  \right\} \cdot (\lambda-z)^2 \\
    + & \left \{ \dot{D}(t,i) + \sum_{j=1}^l q_{ij} D(t,j) - \frac{\xi}{2} \log \left(\frac{\pi \xi}{\sigma^2(t,i) P(t,i)}\right) \right \} = 0,
\end{split}
\end{equation}
indicating that the reduced HJB equation holds. 

We remind the reader that the above equation holds because $P(t,i), H(t,i), C(t,i), D(t,i)$ solves the following system of ODEs
\begin{align}
& \begin{cases} 
    \dot{P}(t,i) = (\rho^2(t,i) - 2r(t,i)) P(t,i) - \sum_{j=1}^l q_{ij}P(t,j) \\
    P(T,i) =1, \text{ for } i \in \{1,\cdots, l\}
\end{cases} \\
& \begin{cases} 
    \dot{H}(t,i) = r(t,i) H(t,i) - \frac{1}{P(t,i)} \sum_{j=1}^l q_{ij} P(t,j) (H(t,j) - H(t,i)) \\
    H(T,i) =1, \text{ for } i \in \{1,\cdots, l\}
\end{cases} \\
& \begin{cases} 
    \dot{C}(t,i) = - \sum_{j=1}^l q_{ij} \left[P(t,j) (H(t,j) - H(t,i))^2 + C(t,j)\right] \\
    C(T,i) = 0, \text{ for } i \in \{1,\cdots, l\}
\end{cases} \\
&\begin{cases} 
    \dot{D}(t,i) = \frac{\xi}{2} \log \left(\frac{\pi \xi}{\sigma^2(t,i) P(t,i)}\right) - \sum_{j=1}^l q_{ij} D(t,j) \\
    D(t,i)= 0, \text{ for } i \in \{1,\cdots, l\}
\end{cases}
\end{align}
\citep{zhou2003markowitz} discussed the existence and uniqueness of the solutions to Eq. \ref{eq:MVRS-P(t,i)-ODE} and \ref{eq:MVRS-H(t,i)-ODE}, which follows since they are linear with uniformly bounded coefficients. Applying Lemma \ref{lm:solving-general-ODE-for-C(t,i)-D(t,i)}, we can explicitly solve the ODEs for $C(t,i)$ and $D(t,i)$:
\begin{align}
\begin{split}
    C(t,i) & = \mathbb{E} \left[\left. \int_t^T \sum_{j=1}^l q_{\alpha(s) j} P(s,j) (H(s,j) - H(s, \alpha(s)))^2 ds \right| \alpha(t) = i \right] \\
    & = \sum_{m=1}^l \sum_{j=1}^l  \int_t^T p_{im}(s-t) q_{mj} P(s,j) (H(s,j) - H(s, m))^2 ds,
\end{split} \\
\begin{split}
    D(t,i) & = - \mathbb{E} \left[ \left. \int_t^T \frac{\xi}{2} \log \left(\frac{\xi\pi}{\sigma(s, \alpha(s)) P(s, \alpha(s))}\right) ds \right| \alpha(t) = i \right] \\
    & = - \sum_{m=1}^l \int_t^T p_{im}(s-t) \frac{\xi}{2} \log \left(\frac{ \pi \xi}{\sigma^2(s, m) P(s, m)}\right) ds.
\end{split}
\end{align}

Lastly, we notice that $V^*$ is indeed smooth as all of its components $P, H, C, D$ are smooth functions. 
\end{proof}

\subsection{Proof of Theorem \ref{thm:EMVRS-PIT}}
\begin{proof}
Consider $0 < t < s < T, X^\pi_t = x, \alpha_t = i$ and suppose $V^\pi$ is smooth. Let $\{X^{\pi^*, t, x}_k\}_{k \in [t,T]}$ be the wealth process that follows policy $\boldsymbol{\pi^*}$ and starts with wealth $x$ at time $t$. Then, by It\^o's Formula, 
\begin{equation} \label{eq:PIT-ito-formula-on-V}
\begin{split}
    & V^\pi(s, X^{\pi^*,t,x}_s, \alpha_s) - V^\pi(t,x,i) \\
    & = \int_t^s \Biggl[ V_t^\pi(k, X^{\pi^*, t, x}_k, \alpha_k) 
    + V_x^\pi(k, X^{\pi^*,t,x}_k, \alpha_k) r(k, \alpha_k) X^{\pi^*, t,x}_k  + \sum_{j=1}^l q_{ij} V^\pi(k, X^{\pi^*,t,x}_k, j) \Biggr] dk \\
    & + \int_t^s \int_{\mathcal{A}} \Biggl[ \frac{1}{2} V_{xx}^\pi (k, X^{\pi^*,t,x}_k, \alpha_k) \sigma^2(k, \alpha_k) u^2 + V_x^\pi(k, X^{\pi^*,t,x}_k, \alpha_k) \rho(k, \alpha_k) \sigma(k, \alpha_k) u \Biggr] \pi_k(u|\alpha_k) du dk.
\end{split}
\end{equation}
As an admissible policy $\boldsymbol{\pi}$, we know that its corresponding {\it value function} $V^\pi$ has the recursive form (Eq. \ref{eq:EMVRS-DPP-recursive-form-of-Value-function}), by the DPP. Under the smoothness assumption on $V^\pi$, we apply It\^o's Formula which yields the following HJB equation
\begin{align*}
    & V_t^\pi (t,x,i) + V_x^\pi(t,x,i) r(t,i) x + \sum_{j=1}^l q_{ij} V^\pi(t,x,i) \\
    & + \int_{\mathcal{A}} \left[\frac{1}{2} V_{xx}^\pi(t,x,i) \sigma^2(t,i) u^2 + V_x^\pi(t,x,i) \rho(t,i) \sigma(t,i) u + \xi \log\pi_t(u|i) \right] \pi_t(u|i) du = 0,
\end{align*}
for $(t,x) \in [0,T] \times \mathbb{R}$ and $i \in \{1, \cdots, l\}$.
Moreover, since $\boldsymbol{\pi^*}$ is constructed from $V^\pi$ through Eq. \ref{eq:PIT-construction-of-pi*}, 
and $\forall (t,x) \in [0,T] \times \mathbb{R}, i \in \{1, \cdots, l\}$,
\begin{align*}
    \pi^*_t(\cdot|i) = \underset{\pi'(\cdot|i) \in \mathcal{A}^\pi}{\arg \min} \int_{\mathcal{A}} \left[\frac{1}{2} V_{xx}^\pi(t,x,i) \sigma^2(t,i) u^2 + V_x^\pi(t,x,i) \rho(t,i) \sigma(t,i) u + \xi \log\pi'(u|i) \right] \pi'(u|i) du.
\end{align*}
Hence, $\forall (t,x) \in [0,T] \times \mathbb{R}, i \in \{1, \cdots, l\}$,
\begin{align*}
    & V_t^\pi (t,x,i) + V_x^\pi(t,x,i) r(t,i) x + \sum_{j=1}^l q_{ij} V^\pi(t,x,i) \\
    & + \int_{\mathcal{A}} \left[\frac{1}{2} V_{xx}^\pi(t,x,i) \sigma^2(t,i) u^2 + V_x^\pi(t,x,i) \rho(t,i) \sigma(t,i) u + \xi \log\pi^*_t(u|i) \right] \pi^*_t(u|i) du \leq 0.
\end{align*}
Substituting this back into Eq \ref{eq:PIT-ito-formula-on-V} yields
\begin{align*}
    & V^\pi(s, X^{\pi^*,t,x}(s), \alpha(s)) - V^\pi(t,x,i) \leq - \xi \int_t^s \int_{\mathcal{A}} \pi^*_k(u|\alpha(k)) \log\pi^*_k(u|\alpha(k)) du dk \\
    \implies & V^\pi(t,x,i) \geq V^\pi(s, X^{\pi^*,t,x}(s), \alpha(s)) + \xi \int_t^s \int_{\mathcal{A}} \pi^*_k(u|\alpha(k)) \log\pi^*_k(u|\alpha(k)) du dk.
\end{align*}
Setting $s = T$ and taking expectation on both sides conditioned on $X_t = x, \alpha_t = i$, we have: 
\begin{align*}
    & \mathbb{E} \left[ V^\pi(t,x,i) | X_t = x, \alpha_t = i \right] \\
    & \geq \mathbb{E}\left[ \left.
    V^\pi(T, X^{\pi^*,t,x}_T, \alpha_T) + \xi \int_t^T \int_{\mathcal{A}} \pi^*_k(u|\alpha_k) \log\pi^*_k(u|\alpha_k) du dk \right| X_t = x, \alpha_t = i \right] \\
    & = \mathbb{E}\left[ \left.
    V^{\pi^*}(T, X^{\pi^*,t,x}_T, \alpha_T) + \xi \int_t^T \int_{\mathcal{A}} \pi^*_k(u|\alpha_k) \log\pi^*_k(u|\alpha_k) du dk \right| X_t = x, \alpha_t = i \right] \\
    & = V^{\pi^*}(t,x,i).
\end{align*}
The second to the last equation is because $V^\pi(T, x, i) = V^{\pi^*}(T,x,i) = (x+(\lambda-z))^2 - \lambda^2$, for all $x \in \mathbb{R}, i \in \{1,\cdots, l\}$, while the last equation is implied by Bellman's Principle of Optimality. Finally, note that the left-hand-side is $V^\pi(t,x,i)$, which concludes the proof. 
\end{proof}

\subsection{Proof of Corollary \ref{thm:EMV-solution}}
\begin{proof}
    By Theorem \ref{thm:EMVRS-solution}, the EMV problem \ref{pb:EMV-problem} 
    has an optimal value function taking the same form as Eq. \ref{eq:EMVRS-optimal-value-function}: 
\begin{equation}
\begin{split}
     V^*(t,x) = P(t)[x+(\lambda-z)H(t)]^2 + (\lambda-z)^2 C(t) + D(t) - \lambda^2,
\end{split}
\end{equation}
where $P(t,i), H(t,i), C(t,i), D(t,i)$ solve the following system of ODEs, which are simplied from Eq. \ref{eq:MVRS-P(t,i)-ODE}, \ref{eq:MVRS-H(t,i)-ODE}, \ref{eq:EMVRS-C(t,i)} and \ref{eq:EMVRS-D(t,i)}: 
\begin{align}
& \begin{cases} 
    \dot{P}(t) = (\rho^2 - 2r) P(t) \\
    P(T) =1, \text{ for } i \in \{1,\cdots, l\}
\end{cases} \\
& \begin{cases} 
    \dot{H}(t) = r(t) H(t) \\
    H(T) =1, \text{ for } i \in \{1,\cdots, l\}
\end{cases} \\
& \begin{cases} 
    \dot{C}(t) = 0\\
    C(T) = 0, \text{ for } i \in \{1,\cdots, l\}
\end{cases} \\
&\begin{cases} 
    \dot{D}(t) = \frac{\xi}{2} \log \left(\frac{\xi \pi}{\sigma^2 P(t)}\right)\\
    D(t)= 0, \text{ for } i \in \{1,\cdots, l\}
\end{cases}
\end{align}
because $\{q_{ij}\}_{i,j = 1,\cdots,l} = 0$ when there is no regime switching. These ODEs can be easily solved: 
\begin{align}
    & P(t) = e^{-(\rho^2-2r)(T-t)}, \\
    & H(t) = e^{-r(T-t)}, \\
    & C(t) = 0, \\
    & D(t) = \frac{\xi (\rho^2-2r)}{4}(T^2-t^2) - \frac{\xi}{2} \left[(\rho^2-2r) T - \log \frac{\sigma^2}{\pi\xi}\right] (T-t).
\end{align}
The optimal policy distribution can be easily derived through Eq. \ref{eq:PIT-construction-of-pi*}, and $\lambda^*$ follows from Eq. \ref{eq:EMVRS-optimal-lambda} in Theorem \ref{thm:EMVRS-solution}.
\end{proof}

\end{document}